\documentclass[11pt]{article}

\usepackage[margin=1in]{geometry}
\usepackage{amsmath, amsthm, amssymb, amsfonts}
\usepackage{hyperref}
\usepackage{enumerate}

\newtheorem{theorem}{Theorem}
\newtheorem{lemma}[theorem]{Lemma}

\DeclareMathOperator*{\argmin}{arg\,min}
\newcommand{\opt}{\text{OPT}}
\newcommand{\sat}{\text{SAT}}
\newcommand{\E}{\mathbb{E}}
\newcommand{\R}{\mathbb{R}}

\newcommand{\DP}{\text{DP}}
\newcommand{\one}{\mathbb{I}}

\title{\bf Deterministic Stackelberg Strategies\\in Robotaxi Assignment Games}

\author{Ioannis Caragiannis\\Aarhus University \&\\Google Research\\\url{iannis@cs.au.dk} \and Kostas Kollias\\Google Research\\\url{kostaskollias@google.com} \and Mohammad Roghani\\Google Research\\\url{roghani@google.com} \and Aaron Schild\\Google Research\\\url{aschild@google.com} \and Ali Kemal Sinop\\Google Research\\\url{asinop@google.com}}

\date{}

\allowdisplaybreaks

\begin{document}

\maketitle

\begin{abstract}
\noindent In autonomous mobility duopolies, platforms compete for riders based on vehicle proximity. A critical operational challenge---faced by an incumbent proactively defending against a market entrant---is how to dispatch fleets to minimize demand poached by a reactive competitor. We model this competition as a zero-sum robotaxi assignment game, focusing on computing optimal deterministic Stackelberg strategies for the leader. We formulate this problem as an exact single-level integer linear program and use its relaxation to design a polynomial-time approximation algorithm achieving an additive error from the optimum. We argue that this is not far from best possible by proving that deciding whether a strictly positive payoff is achievable is NP-complete, and that approximating the optimal payoff within an additive error scaling linearly with the number of riders is NP-hard. We also provide an exact polynomial-time dynamic programming algorithm for computing the optimal Stackelberg payoff when riders and vehicles are on the Euclidean line. Finally, we complement our theoretical findings with an empirical evaluation using real-world ride-hailing data. The experiments demonstrate that our LP-based algorithm significantly outperforms its worst-case guarantee in practice, while lightweight heuristics achieve highly competitive payoffs with negligible computational overhead, offering practical solutions for large-scale deployments.
\end{abstract}

\section{Introduction}
The rapid evolution of autonomous vehicle technology is fundamentally transforming urban transportation, giving rise to commercial robotaxi platforms (such as Waymo, Baidu's Apollo Go, and Zoox). This transition marks a decisive structural shift from traditional ride-hailing services---such as Uber, Lyft, or conventional taxi fleets---which rely on independent human drivers who exercise localized autonomy regarding ride acceptance and routing. In contrast, autonomous mobility services operate under centralized control, allowing platforms to execute globally coordinated, fleet-wide vehicle deployments. The economic viability and long-term prosperity of these capital-intensive platforms critically depend on their ability to secure a large fraction of the rider population across their operating areas. In practice, however, these platforms do not operate in a vacuum; urban markets naturally foster an environment where multiple competing robotaxi companies vie for the very same riders. Consequently, evaluating a firm's true potential to capture and sustain a large customer base crucially hinges on understanding this strategic platform competition.

This competitive dynamic naturally manifests as a sequential interaction, where one platform must commit to its operational deployment before its rival. A quintessential scenario arises when an incumbent firm, already established in an urban market, seeks to proactively assign its fleet to safeguard its customer base against an aggressive incoming entrant. Conversely, an ambitious newcomer entering the market may commit to early service guarantees across key travel corridors, compelling the incumbent to react. In either setting, the first-moving platform must anticipate how a competitor will position vehicles to poach riders. We capture this strategic tension by formulating the interaction as a zero-sum \emph{robotaxi assignment game} played under the {\em Stackelberg} framework. The game features two competing platforms---Firm~1 (the \emph{leader}) and Firm~2 (the \emph{follower})---vying for a shared population of spatially distributed riders with equal-sized fleets of robotaxis. The interaction unfolds in two stages: the leader first commits to a fleet-wide assignment of its vehicles to riders. Observing this deployment, the follower responds with its own vehicle assignment. Riders act deterministically, choosing whichever platform provides the closer vehicle (and thus a shorter wait time). Within this setting, the leader's central goal is to identify an \emph{optimal Stackelberg strategy}---a defensive assignment that minimizes the follower's ability to exploit spatial vulnerabilities and guarantees the largest possible customer base.

Under this sequential model, the central question for the leader is how to compute an \emph{optimal Stackelberg strategy}: a matching that proactively minimizes the follower's capacity to exploit spatial vulnerabilities and guarantees the largest possible customer base. From an abstract game-theoretic perspective, finding an optimal randomized (mixed) Stackelberg strategy in a two-player zero-sum game reduces to computing a minimax equilibrium. However, because each platform's strategy space encompasses all perfect matchings, even formulating this problem compactly is a significant challenge: a naive normal-form representation is exponentially large, and designing a polynomial-sized linear program requires non-trivially encoding the follower's combinatorial response. While we show that this mixed-strategy game can indeed be compactly characterized in polynomial time via our single-level formulation---by coupling the leader's marginal assignments with a dual minimum vertex cover representing the follower's poaching response---deploying randomized strategies remains operationally non-viable for physical fleets. A platform cannot dispatch vehicles according to a probability distribution over millions of potential 
matchings: doing so precludes proactive vehicle pre-positioning, destabilizes urban traffic routing, and introduces unacceptable variance into passenger pickup times. Ensuring operational viability, predictability, and reliable service-level guarantees mandates strict, deterministic commitments.

Therefore, our focus lies squarely on computing optimal \emph{deterministic} (pure) Stackelberg strategies. By enforcing a pure strategy space, the problem departs from the continuous domain and enters the combinatorial one. Identifying these optimal defensive matchings becomes computationally demanding, requiring the leader to anticipate the combinatorial cascade of the follower's discrete best response. Specifically, because the follower solves a maximum-cardinality bipartite matching problem to poach riders, the leader faces an intrinsically discrete, nested $\max\text{-}\min$ bilevel optimization problem: assigning a vehicle to protect one rider inevitably shifts fleet capacity away from another, triggering non-local vulnerabilities across the network. Resolving this combinatorial tension requires overcoming the nested nature of the follower's response and developing dedicated algorithmic tools to compute, approximate, and delineate the complexity boundaries of optimal deterministic commitments.

\subsection{Our contribution and techniques}
Our main technical contribution is a comprehensive study of the computational complexity and approximability of optimal deterministic Stackelberg strategies in robotaxi assignment games.

First, we resolve the nested min-max formulation by introducing the concept of a ``stealing graph.'' For any strategy played by the leader, the stealing graph encodes all possible valid assignments the follower can utilize to provide strictly shorter distances to the riders. By applying K\H{o}nig's theorem to this stealing graph, we translate the follower's maximum-cardinality matching best response into a minimum vertex cover. This equivalence allows us to circumvent the intractable min-max bottleneck and express the computation of an optimal deterministic Stackelberg strategy as a single-level Integer Linear Program (ILP) that maximizes the leader's payoff. In addition, we show that the continuous relaxation of this ILP completely characterizes the game in mixed strategies: its optimal objective value equals the minimax value of the game, and applying the Birkhoff--von Neumann decomposition directly yields an optimal mixed Stackelberg strategy for Firm~1 as well as a mixed Nash equilibrium.

Building upon this relaxation, we develop a polynomial-time approximation algorithm for deterministic commitments. By decomposing the optimal fractional matching into a convex combination of perfect matchings and selecting the best matching in its support, we derandomize the solution. We prove this algorithm guarantees a deterministic payoff within an additive error of at most $n/2$ from the optimum, where $n$ is the number of riders. We complement this by constructing an integrality gap instance on a 1D Euclidean line, demonstrating that the limits of this linear programming relaxation perfectly match our derived approximation guarantee.

Next, we map the computational boundaries of the problem. We prove that deciding whether Firm~1 can achieve a strictly positive Stackelberg payoff is NP-complete via a reduction from \textsc{Max-Clique}, thereby ruling out any finite multiplicative approximation. Furthermore, we establish that a linear additive error is also structurally unavoidable: we utilize a reduction from \textsc{Max-3-Lin}(2) to prove that approximating the optimal Stackelberg payoff within an additive error of $\left(\frac{1}{24} - \delta\right)n$ is NP-hard for any constant $\delta > 0$.

Furthermore, we show that natural geometric structure restores tractability. We present an exact $O(n^3)$ dynamic programming algorithm for riders and robotaxis located on the Euclidean line by strategically partitioning the line at boundaries defined by the follower's robotaxis and leveraging linear-time convex bipartite matching.

Finally, we complement our theoretical results with an empirical evaluation using a real-world ride-hailing trip dataset. The results show that our LP-based rounding algorithm performs substantially better than its worst-case theoretical guarantee. In addition, recognizing the computational bottleneck of solving large linear programs in real time, we evaluate lightweight heuristics and demonstrate that they secure near-optimal Stackelberg payoffs while running orders of magnitude faster, providing practical alternatives for large-scale deployments.

\subsection{Related work}
Our work is situated at the intersection of algorithmic game theory and operations research, drawing upon foundational techniques for computing Stackelberg equilibria. The computational complexity of determining an optimal mixed strategy for a leader to commit to was formalized by Conitzer and Sandholm \cite{conitzer2006computing}, who introduced efficient linear programming formulations for normal-form games. This framework heavily influenced the domain of Stackelberg security games, where discrete resource allocation on bipartite graphs became a central focus \cite{kiekintveld2009computing, korzhyk2010complexity}. Notably, Korzhyk et al.\ \cite{korzhyk2010complexity} demonstrated that computing optimal randomized strategies for security resource allocation can be achieved in polynomial time via fractional linear programming relaxations and the Birkhoff--von Neumann decomposition, whereas computing optimal deterministic commitments is NP-hard. While these prior works leverage this decomposition pipeline to output mixed probability distributions, our approximation algorithm repurposes these exact mathematical tools to extract robust deterministic commitments over a spatial metric.

Furthermore, our integer linear programming formulation builds upon an established methodology in algorithmic game theory and bilevel optimization for capturing discrete follower responses. When a follower's best response involves solving a polynomial-time combinatorial optimization problem, the leader's bilevel problem can often be encoded into a single-level mathematical program by exploiting min-max structural dualities. This approach has been applied in Stackelberg network pricing games \cite{briest2010stackelberg}, where linear programming duality captures the follower's shortest-path decisions, as well as in Stackelberg packing games \cite{bohnlein2023stackelberg}. In our setting, we exploit the classic min-max duality of bipartite matching formalized by K\H{o}nig's theorem, dualizing the follower's maximum-cardinality poaching response into a minimum vertex cover to obtain an exact single-level ILP.

Our setting is also fundamentally connected to classical models of spatial competition, originating with Hotelling's seminal line duopoly \cite{hotelling1929stability}, and geometric facility location models such as Voronoi games \cite{ahn2004competitive, cheong2004one}. In these games, competing firms position facilities across a metric space to maximize their exclusively captured market share based on nearest-neighbor proximity. However, standard spatial location games assume uncapacitated facilities: any facility can serve an unbounded number of customers within its catchment area. In contrast, our fleet deployment imposes rigid unit-capacity constraints on every vehicle, layering perfect bipartite matching onto spatial distance preferences and giving rise to non-local combinatorial trade-offs absent from uncapacitated Voronoi competition.

Finally, our application domain connects to the rapidly expanding operations research literature on game-theoretic competition in ride-hailing networks. Studies in this domain have extensively investigated dynamic spatial pricing and network equilibria \cite{bimpikis2019spatial}, competitive pricing in ride-pooling markets \cite{ke2020pricing}, and multi-homing dynamics between competing platforms \cite{bryan2019theory}. With the advent of autonomous mobility, recent research has begun examining competitive mechanics specifically between centrally controlled robotaxi fleets \cite{siddiq2022ride}. However, these prior works typically analyze competition through macroscopic fluid approximations, stochastic queuing models, or continuous price and wage levers. By contrast, our work models platform competition at a discrete, microscopic algorithmic level, providing an exact mathematical framework for analyzing the direct spatial routing mechanics and fleet assignment duel between competing autonomous mobility platforms.

\subsection{Roadmap}
The rest of the paper is structured as follows. We present preliminary definitions in Section~\ref{sec:setting}. Our LP-based approximation algorithm for optimal Stackelberg payoff is presented in Section~\ref{sec:apx}. Our hardness results follow in Sections~\ref{sec:0vs1} and~\ref{sec:hardness-of-apx}. We present our dynamic programming algorithm for the problem on the Euclidean line in Section~\ref{sec:1D}. The experimental results follow in Section~\ref{sec:exp}. We conclude with open problems in Section~\ref{sec:open}. In Appendix~\ref{sec:mixed-stackelberg}, we show that the LP relaxation used by our approximation algorithm essentially computes the optimal Stackelberg strategy/payoff over randomized strategies.

\section{Preliminaries}\label{sec:setting}
We study the following two-player {\em robotaxi assignment} game. The setting consists of a set $X$ of $n$ riders and two sets $Y$ and $Z$ of $n$ robotaxis each, controlled by the two firms, who act as players. Both the riders and the robotaxis have fixed locations. The function $d:X\times (Y\cup Z)\rightarrow \R_{\geq 0}$ maps each rider-robotaxi pair to a non-negative distance.

The game is defined as follows. Firm 1 has all possible perfect matchings between the riders in $X$ and the robotaxis in $Y$ as strategies. Similarly, Firm 2 has all perfect matchings between the riders in $X$ and the robotaxis in $Z$ as strategies. For a matching $M$ between the riders and one of the two robotaxi sets, the notation $M(i)$ refers to the robotaxi (in $Y$ or $Z$) the rider $i\in X$ is matched to in $M$. Given a pair of strategies $M_1$ and $M_2$ by the two firms, rider $i\in X$ selects to be served by Firm 1 if $d(i,M_1(i))\leq d(i,M_2(i))$ and by Firm 2 otherwise.\footnote{As it is typical in the literature on Stackelberg strategies (e.g., see~\cite{conitzer2006computing,vonstengel2010leadership}), we assume that ties are broken in favor of the leader.} Given a pair of strategies by the firms, the payoff of each firm is the number of riders that select to be served by the firm.

We consider Stackelberg strategies~\cite{von1934marktform}, assuming that Firm 1 acts as the {\em leader} and Firm 2 acts as the {\em follower}. For a strategy $M$ selected by Firm 1, the corresponding {\em Stackelberg payoff} is the minimum payoff of Firm 1 when it selects the strategy $M$, over all possible strategies Firm 2 may select. An optimal Stackelberg strategy is one that maximizes the Stackelberg payoff.

To compute the Stackelberg payoff of a strategy $M$ for Firm 1, it is useful to study how many riders Firm 2 can ``steal'' from Firm 1. We define the {\em stealing graph} $G(M)$ as follows. $G(M)$ is a bipartite graph with the sets $X$ and $Z$ as its two vertex sets. For $i\in X$ and $\ell\in Z$, the edge $(i,\ell)$ exists in $G(M)$ if $d(i,\ell)<d(i,M(i))$. Notice that if Firm 2 selects as a strategy a matching that includes edge $(i,\ell)$, it steals rider $i$ in the sense that they will select to be served by Firm 2. Essentially, the stealing graph $G(M)$ contains those edges that correspond to ways (i.e., matches of riders to robotaxis in $Z$) Firm 2 can steal riders from Firm 1. A maximum-cardinality matching in $G(M)$ yields the maximum number of riders Firm 2 can steal. Clearly, such a matching can be trivially extended to a strategy for Firm 2, i.e., to a perfect matching between $X$ and $Z$, by arbitrarily pairing the remaining unmatched vertices. Given the matching $M$ Firm 1 uses as strategy, we denote the size of this maximum cardinality matching by $\nu(M)$. Then, the Stackelberg payoff of strategy $M$ is $n-\nu(M)$.

We now show how to compute the matching $M$ that minimizes $\nu(M)$ and, thus, yields the optimal Stackelberg payoff. By K\H{o}nig's theorem~\cite{konig1931}, the cardinality $\nu(M)$ of the maximum-cardinality matching in the stealing graph $G(M)$ is equal to the size of the minimum vertex cover in it. Furthermore, for $i\in X$ and $\ell \in Z$, let $D_{i,\ell}=\{j\in Y: d(i,\ell)<d(i,j)\}$. Then, the edge $(i,\ell)$ exists in the stealing graph $G(M)$ if and only if $M(i)\in D_{i,\ell}$. Using the above observations, we can express the problem of computing the optimal Stackelberg strategy and the corresponding payoff as the following integer linear program:
\begin{align}\tag{ILP 1}\label{eq:ilp}
\mbox{maximize\ \ } & n-\sum_{i\in X}{u_i}-\sum_{\ell \in Z}{v_\ell} &\\\label{eq:ilp-matching-constraint-1}
\mbox{subject to:\ \ } & \sum_{j\in Y}{x_{i,j}} = 1, &\forall i\in X\\\label{eq:ilp-matching-constraint-2}
& \sum_{i\in X}{x_{i,j}} = 1, &\forall j\in Y\\\label{eq:ilp-vertex-cover-constraint}
& u_i+v_\ell \geq \sum_{j\in D_{i,\ell}}{x_{i,j}}, &\forall i\in X, \ell \in Z\\\nonumber
& x_{i,j} \in \{0,1\}, &\forall i\in X, j\in Y\\\nonumber
& u_i\in \{0,1\}, &\forall i\in X\\\nonumber
& v_\ell \in \{0,1\}, &\forall \ell \in Z
\end{align}
The matrix $x$ represents a perfect matching $M$. For $i\in X$ and $j\in Y$, the variable $x_{i,j}$ denotes whether the edge $(i,j)$ appears in the matching $M$ ($x_{i,j}=1$) or not ($x_{i,j}=0$). Equations (\ref{eq:ilp-matching-constraint-1}) and (\ref{eq:ilp-matching-constraint-2}) are the perfect matching constraints requiring that each vertex $i\in X$ and $j\in Y$ should be matched in $M$, respectively. The vectors $u$ and $v$ represent the vertex cover in $G(M)$. The variable $u_i$ indicates whether vertex $i\in X$ belongs to the vertex cover ($u_i=1$) or not ($u_i=0$). The variable $v_\ell$ indicates whether vertex $\ell\in Z$ belongs to the vertex cover ($v_\ell=1$) or not ($v_\ell=0$). 

Due to the constraint (\ref{eq:ilp-matching-constraint-1}), the quantity $\sum_{j\in D_{i,\ell}}{x_{i,j}}$ is either $0$ or $1$. It is equal to $1$ if and only if rider $i\in X$ is matched to a robotaxi in $D_{i,\ell}$ in $M$ and, hence, the edge $(i,\ell)$ belongs to the stealing graph $G(M)$. In this case, the vertex cover should include either vertex $i\in X$ or vertex $\ell\in Z$; this is captured by constraint (\ref{eq:ilp-vertex-cover-constraint}). The last three constraints of the integer linear program \ref{eq:ilp} are the integrality constraints for the variables.

Before addressing the challenge of computing high-payoff deterministic Stackelberg strategies, we observe that the continuous relaxation of \ref{eq:ilp} (obtained by replacing the integrality constraints with non-negativity constraints) completely characterizes the robotaxi assignment game when players are permitted to randomize over matchings. In Appendix~\ref{sec:mixed-stackelberg}, we show for completeness that its optimal objective value $\mathrm{OPT}_{\mathrm{LP}}$ exactly equals the value of the zero-sum assignment game in mixed strategies. While computing commitment strategies in general-sum games requires solving multiple linear programs~\cite{conitzer2006computing}, in two-player zero-sum games a classic result establishes that the leader cannot obtain more than the minimax value by committing to a mixed strategy~\cite{korzhyk2011stackelberg, vonstengel2010leadership}. By linear programming duality, the minimax value and optimal commitment strategies can be computed efficiently via a single linear program~\cite{adler2013equivalence, dantzig1951equivalence}. Consequently, solving the continuous relaxation of \ref{eq:ilp} and applying the Birkhoff--von Neumann decomposition directly yields an optimal mixed Stackelberg strategy for Firm~1 and, equivalently, a mixed Nash equilibrium.

\section{An approximation algorithm}\label{sec:apx}
We begin the exposition of our technical results by focusing on computing a deterministic strategy for the leader with high Stackelberg payoff. We analyze an approximation algorithm that rounds the optimal fractional solution $(x^*, u^*, v^*)$ of the linear programming relaxation of \ref{eq:ilp} into a pure matching. The algorithm runs in three steps:
\begin{enumerate}
    \item Solve the linear programming relaxation of \ref{eq:ilp} and compute an optimal fractional solution $(x^*,u^*,v^*)$.
    \item Decompose the fractional matching $x^*$ into a convex combination of perfect matchings forming the collection of matchings $\mathcal{M}$.
    \item Return $\argmin_{M\in \mathcal{M}}{\nu(M)}$.
\end{enumerate}
We will prove the following statement.
\begin{theorem}\label{thm:stackelberg}
The algorithm runs in polynomial time and returns a matching $M^*$ that guarantees a Stackelberg payoff of
\begin{align}\label{eq:apx-guarantee}
    n-\nu(M^*) &\geq n\cdot \left(1-\sqrt{\frac{2\cdot (n-\opt)}{n+1}}\right)
\end{align}
for the leader, where $\opt$ denotes their optimal Stackelberg payoff.
\end{theorem}

Proving that the algorithm runs in polynomial time is easy. Solving the linear programming relaxation can be clearly done in polynomial time; e.g., see~\cite{schrijver1986theory}. The fractional matching $x^*$ is essentially a doubly stochastic $n \times n$ matrix, which can be decomposed into $O(n^2)$ permutation matrices, using the Birkhoff--von Neumann theorem~\cite{birkhoff1946tres,von1953certain}. Each permutation matrix defines a perfect matching in $\mathcal{M}$. Finally, computing $\nu(M)$ for a given matching $M$ can be done by computing a maximum cardinality matching (e.g., see~\cite{cormen2009introduction}) in the stealing graph $G(M)$.

To prove the Stackelberg payoff guarantee, we will use a probabilistic argument. We will analyze the following randomized version of the algorithm.
\begin{enumerate}
    \item Solve the linear programming relaxation of \ref{eq:ilp} and compute an optimal fractional solution $(x^*,u^*,v^*)$.
    \item Decompose the fractional matching $x^*$ into a convex combination of perfect matchings forming the collection of matchings $\mathcal{M}$. 
    \item Use the coefficients of the convex combination to define a probability distribution $P$ over the matchings of $\mathcal{M}$ such that $\Pr_{M\sim P}[M(i)=j]=x^*_{i,j}$ for every $i\in X$ and $j\in Y$.
    \item Return a random matching from $\mathcal{M}$ drawn according to $P$.
\end{enumerate}
We will show that 
\begin{align}\label{eq:rand-guarantee}
    n-\E_{M\sim P}[\nu(M)] &\geq n\cdot \left(1-\sqrt{\frac{2\cdot (n-\opt_{\text{LP}})}{n+1}}\right),
\end{align}
where $\opt_{\text{LP}}$ denotes the optimal objective value of the linear programming relaxation of integer linear program \ref{eq:ilp}. The Stackelberg payoff guarantee in Theorem~\ref{thm:stackelberg} will then follow since $\nu(M^*)\leq \E_{M\sim P}[\nu(M)]$ and $\opt_{\text{LP}}\geq \opt$.

Our proof of Equation (\ref{eq:rand-guarantee}) and, consequently, of Theorem~\ref{thm:stackelberg} will follow immediately by the next lemma.

\begin{lemma} \label{lem:upper_bound}
The optimal solution $(x^*, u^*, v^*)$ of the linear programming relaxation  satisfies
\begin{align*}
    \opt_{\text{LP}} & \leq n - \frac{(n+1)}{2n^2}\cdot \E_{M\sim P}[\nu(M)]^2.
\end{align*}
\end{lemma}

\begin{proof}
Let $(x^*, u^*, v^*)$ be the optimal solution of the linear programming relaxation of \ref{eq:ilp}. Consider the next linear program
\begin{align}\tag{LP 2}\label{eq:lp-2}
    \mbox{minimize\ \ } & \sum_{i \in X} u_i + \sum_{\ell \in Z} v_\ell \\\nonumber
    \text{subject to:\ \ } & u_i + v_\ell \geq \sum_{j\in D_{i,\ell}}{x^*_{i,j}}, &\forall i \in X, \ell \in Z \\\nonumber
    & u_i \geq 0, &\forall i \in X \\\nonumber
    & v_\ell \geq 0, &\forall \ell \in Z
\end{align}
which uses the entries of $x^*$ as constants. Clearly, the solution $(u^*, v^*)$ is an optimal solution to \ref{eq:lp-2}. We will show that its objective value is at least $\frac{(n+1)}{2n^2}\cdot \E_{M\sim P}[\nu(M)]^2$, implying that the optimal objective value $\opt_{\text{LP}}$ of the linear programming relaxation of \ref{eq:ilp} is as claimed in the statement of the lemma.

We do so by exploiting duality and proving that the dual linear program defined as
\begin{align} \tag{LP 3}\label{eq:lp-3}
    \text{maximize\ \ } & \sum_{i \in X} \sum_{\ell \in Z} f_{i,\ell} \cdot \sum_{j\in D_{i,\ell}}{x^*_{i,j}} \\\nonumber
    \text{subject to:\ \ } & \sum_{\ell \in Z} f_{i,\ell} \leq 1, &\forall i \in X \\\nonumber
    & \sum_{i \in X} f_{i,\ell} \leq 1, &\forall \ell \in Z \\\nonumber
    & f_{i,\ell} \geq 0, &\forall i \in X, \ell \in Z
\end{align}
has a feasible solution $f^*$ with objective value at least $\frac{(n+1)}{2n^2}\cdot \E_{M\sim P}[\nu(M)]^2$.

For each matching $M$ that Firm 1 can use as strategy, we denote by $F_M$ a
maximum-cardinality matching in the stealing graph $G(M)$ which represents a best response by Firm 2 and thereby defines the Stackelberg payoff of strategy $M$ for Firm 1.\footnote{If more than one such strategy exists for Firm 2, we select one arbitrarily.} For every $i \in X$ and $\ell \in Z$, set 
\begin{align*}
f^*_{i,\ell} &= \Pr_{M \sim P}\left[(i,\ell) \in F_M\right].
\end{align*}
Clearly, $f^*_{i,\ell} \geq 0$. Then, the quantities $\sum_{\ell \in Z} f^*_{i,\ell}$ and $\sum_{i \in X} f^*_{i,\ell}$ are simply the probabilities that rider $i \in X$ is stolen in $F_M$ by some robotaxi in $Z$, and that robotaxi $\ell \in Z$ steals some rider in $F_M$, respectively. Hence, $\sum_{\ell \in Z} f^*_{i,\ell} \leq 1$ and $\sum_{i \in X} f^*_{i,\ell} \leq 1$, implying that $f^*$ is a feasible solution for the dual~\ref{eq:lp-3}.

For rider $i \in X$ and integer $k \in [n]$, we denote by $r_k^i$ the $k$-th furthest robotaxi in $Y$ from rider $i$ (breaking distance ties arbitrarily). We now have that for every $i \in X$ and $\ell \in Z$, it holds that:
\begin{align}
    f_{i,\ell}^* \cdot \sum_{j\in D_{i,\ell}}{x^*_{i,j}} &= \Pr_{M \sim P}\left[(i,\ell) \in F_M\right] \cdot \sum_{j \in D_{i,\ell}}{x^*_{i,j}} \nonumber \\
    &= \sum_{k=1}^n \Pr_{M \sim P}\left[(i,\ell) \in F_M \wedge M(i) = r_k^i\right] \cdot \sum_{j=1}^{|D_{i,\ell}|}{x^*_{i, r_j^i}}. \label{eq:obvious_step}
\end{align}
Equation (\ref{eq:obvious_step}) holds because $\Pr_{M \sim P}[(i,\ell) \in F_M \wedge M(i) = r_k^i] = 0$ if $k > |D_{i,\ell}|$. Because the term $\Pr_{M \sim P}[(i,\ell) \in F_M \wedge M(i) = r_k^i]$ can be non-zero only if $k \leq |D_{i,\ell}|$, we can bound the inner sum by truncating it at the $k$-th term, i.e., 
\begin{align} \nonumber
    f^*_{i,\ell} \cdot \sum_{j\in D_{i,\ell}} x^*_{i, j} &\geq \sum_{k=1}^n{\Pr_{M \sim P}\left[(i,\ell) \in F_M \wedge M(i) = r_k^i\right]} \cdot \sum_{j=1}^{k}{x^*_{i, r_j^i}} \\\label{eq:f-star-obj-value-bound}
    &\geq \sum_{k=1}^n{\Pr_{M \sim P}\left[(i,\ell) \in F_M \wedge M(i) = r_k^i\right] \cdot \sum_{j=1}^{k}{\sum_{\ell' \in Z}{\Pr_{M \sim P}[(i,\ell') \in F_M \wedge M(i) = r_j^i]}}}
\end{align}
The second inequality in the above derivation follows since $x^*_{i, r_j^i}$ is the probability that rider $i$ is assigned to robotaxi $r_j^i \in Y$, which is clearly not lower than the probability that rider $i$ is assigned to robotaxi $r_j^i \in Y$ {\em and} is stolen by some robotaxi $\ell' \in Z$ in $F_M$.

For $i \in X$ and $k \in [n]$, abbreviate as $y_{i,k}$ the probability $\sum_{\ell' \in Z}{\Pr_{M \sim P}[(i,\ell') \in F_M \wedge M(i) = r_k^i]}$
that rider $i$ is assigned to robotaxi $r_k^i \in Y$ in $M$ and is stolen by some robotaxi $\ell' \in Z$ in $F_M$.

Then, using inequality (\ref{eq:f-star-obj-value-bound}), we have that the objective value of the feasible solution $f^*$ for \ref{eq:lp-3} is:
\begin{align}\nonumber
    &\sum_{i \in X}{\sum_{\ell \in Z}{f^*_{i,\ell} \cdot \sum_{j\in D_{i,\ell}}{x^*_{i,j}}}} \\\nonumber
    &\geq \sum_{i \in X} \sum_{\ell \in Z} \sum_{k=1}^n \Pr_{M \sim P}\left[(i,\ell) \in F_M \wedge M(i) = r_k^i\right] \cdot \sum_{j=1}^k \sum_{\ell' \in Z} \Pr_{M \sim P}\left[(i,\ell') \in F_M \wedge M(i) = r_j^i\right] \\\nonumber
    &= \sum_{i \in X} \sum_{k=1}^n \left( \sum_{\ell \in Z} \Pr_{M \sim P}\left[(i,\ell) \in F_M \wedge M(i) = r_k^i\right] \right) \cdot \sum_{j=1}^k \left( \sum_{\ell' \in Z} \Pr_{M \sim P}\left[(i,\ell') \in F_M \wedge M(i) = r_j^i\right] \right) \\\label{eq:dual-obj-value-bound-in-terms-of-ys}
    &= \sum_{i \in X} \sum_{k=1}^n y_{i,k} \sum_{j=1}^k y_{i,j},
\end{align}
where the first equality follows by exchanging the second and third sums, and the second one follows by the definition of $y_{i,k}$.

Notice that for every $i\in X$, we have
\begin{equation}\label{eq:ineq-on-ys}
    \sum_{k=1}^n y_{i,k} \sum_{j=1}^k y_{i,j} = \frac{1}{2} \cdot \left( \left( \sum_{k=1}^n y_{i,k} \right)^2 + \sum_{k=1}^n y_{i,k}^2 \right) \geq \frac{n+1}{2n} \cdot \left( \sum_{k=1}^n y_{i,k} \right)^2.
\end{equation}
The inequality follows by applying the Cauchy-Schwarz inequality, which yields that $\sum_{k=1}^n y_{i,k}^2 \geq \frac{1}{n}\cdot \left(\sum_{k=1}^n y_{i,k}\right)^2$. Hence, by combining inequalities (\ref{eq:dual-obj-value-bound-in-terms-of-ys}) and (\ref{eq:ineq-on-ys}), we obtain
\begin{align*}
    \sum_{i \in X}{\sum_{\ell \in Z}{f^*_{i,\ell} \cdot \sum_{j\in D_{i,\ell}}{x^*_{i,j}}}} &\geq \frac{n+1}{2n}\cdot \sum_{i \in X}{\left(\sum_{k=1}^n{y_{i,k}} \right)^2}\\
    &\geq \frac{n+1}{2n^2}\cdot \left(\sum_{i \in X}{\sum_{k=1}^n{y_{i,k}} }\right)^2\\
    &= \frac{n+1}{2n^2}\cdot \E_{M\sim P}[\nu(M)]^2,
\end{align*}
completing the proof. Again, the second inequality follows by applying the Cauchy-Schwarz inequality, which yields $\sum_{i \in X} \left(\sum_{k=1}^n y_{i,k}\right)^2 \geq \frac{1}{n}\cdot  \left(\sum_{i \in X} \sum_{k=1}^n y_{i,k}\right)^2$. The equality follows from the definitions of $\nu(M)$ and the $y_{i,k}$'s. 
\end{proof}

\subsection{The integrality gap of the linear programming relaxation}
We now show that the guarantee in terms of the optimal objective value of the linear programming relaxation of \ref{eq:ilp} cannot be improved.

\begin{lemma}
For every $k\in \{0,1,2, ..., n\}$, there exists a robotaxi assignment game with $n$ riders, two firms with $n$ robotaxis each, and objective value $\opt_{\text{LP}}(k)$ to the linear programming relaxation of \ref{eq:ilp}, in which every strategy $M$ of Firm 1 yields Stackelberg payoff satisfying
\begin{align*}
    n-\nu(M) &\leq k \leq n\cdot \left(1-\sqrt{\frac{2\cdot (n-\opt_{\text{LP}}(k))}{n+1}}\right)+\frac{1}{2}.
\end{align*}
\end{lemma}

\begin{proof}
Let $k\in \{0, 1, ..., n\}$ and consider the robotaxi assignment game with riders and robotaxis at a Euclidean line with $X=Y=Z=[n]$. All $n$ riders are located at point $0$. Firm 1 has $k$ robotaxis located at point $1$, and one robotaxi in each of the odd-coordinate locations $3$, $5$, ..., and $2n-2k+1$. For $\ell\in [n]$, Firm 2 has robotaxi $\ell$ at location $2\ell$. Clearly, for every matching $M$ of the robotaxis of Firm 1 to riders, Firm 2 can respond by assigning the rider who is assigned to robotaxi at location $2j+1$ by Firm 1 to its robotaxi at location $2j$ for $j=1, 2, ..., n-k$. In this way, Firm 2 steals $n-k$ riders from Firm 1, for a Stackelberg payoff of (at most) $k$.

Now consider the fractional solution to the linear programming relaxation of \ref{eq:ilp}. We set $x^*_{i,j}=1/n$ for every $i\in X$ and $j\in Y$, $u^*_i=0$ and $v^*_{\ell}=\frac{n-k-\ell+1}{n}$ for $\ell=1, 2, ..., n-k$ and $v^*_\ell=0$ for $\ell=n-k+1, ..., n$. Recall that Firm 1 has exactly $k+\ell-1$ robotaxis (specifically, the $k$ robotaxis at location 1 and the $\ell-1$ robotaxis at odd locations $3, 5, ..., 2\ell-1$) that are strictly closer to the riders at point 0 than the robotaxi $\ell$ of Firm 2 at location $2\ell$. Therefore, the robotaxi $\ell$ cannot steal riders assigned to these $k+\ell-1$ robotaxis, but it can steal riders assigned to the remaining $n-k-\ell+1$ robotaxis of Firm 1. Hence, $|D_{i,\ell}| = n-k-\ell+1$ and the vertex cover constraint becomes $u_i+v_\ell \geq \frac{n-k-\ell+1}{n}$, which is satisfied by the definition of the solution $(x^*,u^*,v^*)$ which is, therefore, feasible. The optimal objective value of the linear programming relaxation of \ref{eq:ilp} is then
\begin{align}\label{eq:counter-example-bound}
    \opt_{\text{LP}}(k) &\geq n-\sum_{i\in X}{u^*_i}-\sum_{\ell\in Z}{v^*_\ell} = n-\sum_{\ell=1}^{n-k}{\frac{n-k-\ell+1}{n}}=n-\frac{(n-k)(n-k+1)}{2n}.
\end{align}
Tedious calculations yield that
\begin{align*}
    n\, (n-k) (n-k+1) &\leq (n+1)\left(n-k+\frac{1}{2}\right)^2
\end{align*}
for any non-negative integers $n$ and $k$ with $k\leq n$. Hence, Equation~(\ref{eq:counter-example-bound}) yields
\begin{align*}
    \opt_{\text{LP}}(k) &\geq n-\frac{(n+1)\left(n-k+\frac{1}{2}\right)^2}{2n^2}.
\end{align*}
By rearranging, we obtain the inequality in the statement, completing the proof of the lemma.
\end{proof}

\subsection{Deriving an additive approximation guarantee}
Theorem~\ref{thm:stackelberg} provides a non-standard approximation guarantee. Ideally, we would like to get either a multiplicative guarantee $n-\nu(M^*) \geq \alpha(n) \cdot \opt$ for some function $\alpha(n)$ or an additive approximation guarantee $n-\nu(M^*) \geq \opt-\beta(n)$ for some function $\beta(n)$. Notice that for instances with $\opt \leq \frac{n-1}{2}$ the RHS of inequality (\ref{eq:apx-guarantee}) is non-positive. Hence, our LP-based algorithm presented above achieves no finite multiplicative approximation guarantee.

However, inequality (\ref{eq:apx-guarantee}) provides an approximation guarantee with an additive error at most $n/2$ (or an $n/2$-additive approximation guarantee) via the following derivation:
\begin{align*}
n - \nu(M^*) &\geq n-\sqrt{\frac{2n^2(n-\opt)}{n+1}}\\
&= \opt +\left(\sqrt{n-\opt}-\sqrt{\frac{n^2}{2(n+1)}}\right)^2-\frac{n^2}{2(n+1)} \geq \opt-\frac{n}{2}.
\end{align*}
We devote the next two sections to showing that this is not far from best possible. 

\section{Deciding whether positive Stackelberg payoff is possible is hard}\label{sec:0vs1}
In this section, we show that there is no finite multiplicative guarantee for the problem of computing the optimal Stackelberg payoff.

\begin{theorem}
Given a robotaxi assignment game, it is NP-hard to decide whether there exists a deterministic strategy for Firm 1 that yields a strictly positive Stackelberg payoff. 
\end{theorem}

\begin{proof}
We prove the theorem by presenting a polynomial-time reduction from the \textsc{Maximum Clique} problem, which asks whether a given graph contains a clique of a certain size $c$. This problem is well-known to be NP-hard~\cite{karp1972reducibility} and remains so even when restricted to even integers $c \geq 4$. 

Let $H = (V, E)$ be an instance of \textsc{Maximum Clique}, and let $c \geq 4$ be an even integer. Let $N = |V|$, $m = |E|$, and $k=\binom{c}{2}$. Without loss of generality, we may assume that $N\geq 6$, $N \geq c + 2$, and $m\geq k$ (otherwise, if $N\leq 5$ or $c\geq N-1$ or $m<k$, the \textsc{Maximum Clique} instance can be trivially decided in polynomial time). Given $H$ and $c$, we construct a robotaxi assignment game as follows. Let
\[
D = N\left(\frac{c}{2} - 1\right) + \frac{c}{2} - 1 - k.
\]
By our assumptions on $N$ and $c$, we have $D > 0$. The number of riders and fleet size for both firms is $n = m + D$. The game has a set $X_E$ of {\em edge riders}, with one rider for each edge $e\in E$ and a set $X_D$ containing $D$ {\em dummy riders}. Firm 1's fleet has a set $Y_G$ of $m-k$ {\em good robotaxis} and a set $Y_B$ of $k+D$ {\em bad robotaxis}. Firm 2's fleet consists of three sets $Z_1$, $Z_2$, and $Z_3$ of robotaxis. Set $Z_1$ has $\frac{c}{2}-1$ {\em vertex robotaxis} for each vertex $v\in V$. Set $Z_2$ has $m-k$ {\em protected-stealer robotaxis} while set $Z_3$ has $\frac{c}{2}-1$ {\em auxiliary robotaxis}. The distances between riders and robotaxis of Firm 1 are as follows:
\begin{itemize}
    \item The distance between an edge rider $x_e \in X_E$ and a good robotaxi $y_g \in Y_G$ is $2$.
    \item For all other pairs (edge riders to bad robotaxis, and dummy riders to any robotaxi), the distance is $3$.
\end{itemize}
Now, consider a perfect matching among riders and robotaxis of Firm 1. Since dummy riders are at distance $3$ from all Firm 1 robotaxis, the distance to their assigned Firm 1 robotaxi will always be at least $3$. Firm 1 will assign its $m-k$ good robotaxis to $m-k$ edge riders, giving them the lowest possible distance of $2$ to their assigned Firm 1 robotaxi. The remaining $k$ edge riders receive bad robotaxis (for a distance of $3$); let $S \subseteq X_E$ denote the set of these $k$ edge riders.

The distances between riders and robotaxis of Firm 2 are as follows:\footnote{We remark that all distances between riders and robotaxis belong to $\{1,2,3\}$. Since the graph between riders and robotaxis is bipartite, these distances are consistent with a shortest-path metric on this graph. Of course, the distances among riders or among robotaxis do not affect our arguments in any way.}
\begin{itemize}
    \item The distance between a vertex robotaxi $z\in Z_1$ and an edge rider $x_e$ corresponding to edge $e = (u,v)$ is $2$ if $z$ is associated with vertex $u$ or $v$; otherwise, the distance is $3$. The distance between a vertex robotaxi $z\in Z_1$ and a dummy rider $x_d \in X_D$ is $1$.
    \item The distance between a protected-stealer robotaxi $z\in Z_2$ and any edge rider $x_e\in X_E$ is $1$. The distance between a protected-stealer robotaxi $z\in Z_2$ and any dummy rider $x_d\in X_D$ is $3$.
    \item The distance between an auxiliary robotaxi $z\in Z_3$ and any edge rider $x_e\in X_E$ is $2$. The distance between an auxiliary robotaxi $z\in Z_3$ and any dummy rider $x_d\in X_D$ is $2$. 
\end{itemize}

Given a perfect matching $M$ selected by Firm 1, we define the stealing graph $G(M)$ by considering the edges between riders and robotaxis of Firm 2 that have distance strictly smaller than the distance to the assigned Firm 1 robotaxi. So, the graph $G(M)$ has:
\begin{itemize}
    \item All edges between protected-stealer robotaxis in $Z_2$ and any edge rider in $X_E$.
    \item The edges between the vertex and auxiliary robotaxis in $Z_1\cup Z_3$ and the $k$ vulnerable edge riders in $S$ and the edges between the vertex and auxiliary robotaxis in $Z_1\cup Z_3$ and the dummy riders in $X_D$.
\end{itemize}
Notice that $|Z_1\cup Z_3| = k + D$. To find a maximum-cardinality matching in the stealing graph $G(M)$, Firm 2 will optimally match $Z_2$ to $m-k$ protected edge riders, leaving the $k$ vulnerable edge riders in $S$ and the dummy riders in $X_D$ available to be stolen by $Z_1\cup Z_3$. The total number of stolen riders is thus $(m - k) + \mu(S \cup X_D)$, where $\mu(S\cup X_D)$ denotes the maximum number of matches between the edge riders in $S$ and the dummy riders in $X_D$, and the robotaxis in $Z_1\cup Z_3$. Hence, Firm 1 maximizes its Stackelberg payoff by choosing the set $S$ that minimizes $\mu(S \cup X_D)$.

We are done with the description of the reduction and ready to prove its correctness. First, suppose $H$ contains a clique of size $c$. Firm 1 chooses a matching so that the vulnerable set $S$ consists of the edge riders corresponding to the $k = \binom{c}{2}$ edges of $H$ in this clique. 
Consequently, the active neighborhood of $S$ in $Z_1\cup Z_3$ consists only of the vertex robotaxis in $Z_1$ corresponding to the $c$ vertices forming the clique in $H$, plus all robotaxis in $Z_3$. The size of this neighborhood is:
\begin{align}\label{eq:deficiency-of-S}
|N_{Z_1\cup Z_3}(S)| = c\left(\frac{c}{2} - 1\right) + \left(\frac{c}{2} - 1\right) = k - 1<|S|.
\end{align}
Firm 2 is attempting to match the $k+D$ riders in $S \cup X_D$ into $Z_1\cup Z_3$, which also has size $k+D$. By Equation~(\ref{eq:deficiency-of-S}), Hall's matching theorem~\cite{hall1935representatives} implies that at most $|Z_1\cup Z_3| - 1 = k+D - 1$ riders can be matched. Hence, at least one rider will survive and the Stackelberg payoff of Firm 1 will be strictly positive.

Next, suppose $H$ does not contain a clique of size $c$. Firm 1 is forced to choose \textit{some} $k$-sized subset $S$ of edge riders to be vulnerable. We claim that Firm 2 can completely match $S \cup X_D$ into $Z_1\cup Z_3$, implying that it can steal all riders from Firm 1. Let $A \subseteq S \cup X_D$ be an arbitrary non-empty subset of vulnerable riders. By Hall's matching theorem~\cite{hall1935representatives}, it suffices to show that $|N_{Z_1\cup Z_3}(A)| \geq |A|$. We do so by distinguishing between a few cases.

    \paragraph{Case 1: $A \cap X_D \neq \emptyset$.} Since every dummy rider connects to all robotaxis in $Z_1\cup Z_3$, we have $N_{Z_1\cup Z_3}(A) = Z_1\cup Z_3$ and, thus, $|N_{Z_1\cup Z_3}(A)|=k+D=|S\cup X_D|\geq |A|$, as desired.
    
    \paragraph{Case 2: $A \subseteq S$.} Then $A$ is a set of edge riders. Let $y \geq 1$ be the number of vertices spanned by the edges in $A$. The neighborhood of $A$ in $Z_1\cup Z_3$ consists of the robotaxis in $Z_1$ corresponding to these $y$ vertices, plus all the auxiliary robotaxis in $Z_3$. Its size is:
    \[
    |N_{Z_1\cup Z_3}(A)| = y\left(\frac{c}{2} - 1\right) + \frac{c}{2} - 1 = (y+1)\left(\frac{c}{2} - 1\right)
    \]
    We verify the condition $|N_{Z_1\cup Z_3}(A)|\geq |A|$ for all possible values of $y$:
    \medskip
    
    \noindent
    \textit{Subcase 2.1: $y \leq c-1$.} We have
        \begin{align*}
            |N_{Z_1\cup Z_3}(A)| &= (y+1)\left(\frac{c}{2}-1\right) > \binom{y}{2}\geq |A|,
        \end{align*}
        as desired. The first inequality is equivalent to $y(y-c+1)-c+2<0$ which clearly holds for $0\leq y\leq c-1$ and $c\geq 4$. The second inequality follows since the maximum number of edges among $y$ vertices is $\binom{y}{2}$.   
    \medskip
    
    \noindent
    \textit{Subcase 2.2: $y = c$.} We have 
        \begin{align*}
            |N_{Z_1\cup Z_3}(A)| &= (c+1)\left(\frac{c}{2}-1\right) =\binom{c}{2}-1\geq |A|,
        \end{align*}
        as desired. The inequality follows since the graph $H$ contains no clique of size $c$.
        \medskip
    
    \noindent
    \textit{Subcase 2.3: $y \geq c+1$.} We have
        \begin{align*}
            |N_{Z_1\cup Z_3}(A)| &\geq (c+2)\left(\frac{c}{2}-1\right) \geq\binom{c}{2}=k\geq |A|,
        \end{align*}
        as desired. The second inequality follows by the fact $c\geq 4$ and the third one by the assumption $A\subseteq S$ and the fact $|S|=k$.
\end{proof}

\section{Hardness of approximating the optimal Stackelberg payoff}\label{sec:hardness-of-apx}
Recall that our LP-based algorithm computes a deterministic strategy that approximates the optimal Stackelberg payoff with a modest additive term of $n/2$. We now show that such an $\Omega(n)$ additive error in the approximation is necessary.

\begin{theorem}\label{thm:n_over_24}
For every constant $\delta \in \left(0, \frac{1}{24}\right)$, it is NP-hard to approximate\footnote{In the statement of Theorem~\ref{thm:n_over_24}, we consider specifically algorithms that return the Stackelberg payoff that is associated with a feasible strategy of Firm 1. Hence, they always yield a payoff that can only be lower than the optimal Stackelberg payoff, incurring a one-sided error.} the optimal Stackelberg payoff of an $n$-robotaxi assignment game within an additive error of $\left(\frac{1}{24} - \delta\right) n$.
\end{theorem}

We define a polynomial-time reduction from problem~\textsc{Max-3-Lin(2)}. An instance $H=(V,E)$ of \textsc{Max-3-Lin(2)} consists of a set of $N$ Boolean variables $V=\{x_1, x_2, ..., x_N\}$ and a set $E$ of $m$ clauses $E=\{C_1, ..., C_m\}$. Each clause $e=(u,v,w, b_e) \in E$ represents the linear equation modulo $2$
\[
u \oplus v \oplus w = b_e \pmod 2,
\]
where $u, v,$ and $w$ are three distinct variables from $V$ and $b_e \in \{0, 1\}$. We denote by $\opt(H)$ the maximum number of clauses in $E$ that can be simultaneously satisfied (by some Boolean assignment to the variables of $V$).

We will use the following hardness of approximation result.

\begin{theorem}[H\aa stad~\cite{hastad2001some}]\label{thm:hastad}
For any constant $\varepsilon \in \left(0,1/4\right)$, given a \textsc{Max-3-Lin(2)} instance $H$ with $m$ clauses, it is \textsf{NP}-hard to decide whether $\opt(H)$ is at least $(1-\varepsilon)m$ or at most $\left(\frac{1}{2} + \varepsilon\right)m$.
\end{theorem}

We are now ready to define our reduction. First, for a positive integer $k$, given a \textsc{Max-3-Lin(2)} instance $H=(V,E)$, we construct the $k$-variable blow-up instance $H^k=(V_k, E_k)$ as follows:
\begin{itemize}
    \item Each variable $v \in V$ is replaced by a cloud $C_v$ of $k$ variable copies. Thus, $V_k = \bigcup_{v \in V} C_v$ and $|V_k| = n_k = k \cdot N$.
    \item Each clause $e=(u,v,w,b_e)$ is replaced by the $k^3$ clauses $(u', v', w', b_e)$ for all triplets $(u', v', w')$ of variables from $C_u \times C_v \times C_w$. Thus, $|E_k| = m_k = k^3 m$.
\end{itemize}
For each variable $v \in V_k$, denote by $d_v$ the number of clauses in $E_k$ that contain $v$. Clearly, $\sum_{v \in V_k} d_v = 3m_k$.

Next, we set $n = 12m_k + 4n_k$ and construct the $n$-robotaxi assignment game $\Gamma(H,k)$ as follows:
\begin{itemize}
    \item For each Boolean variable $v \in V_k$, we introduce two sets $X_{v,0}$ and $X_{v,1}$ of $2(d_v + 1)$ {\em variable riders} each. Hence, the total number of riders is 
    \[
    |X| = \sum_{v \in V_k} 4(d_v + 1) = 12m_k + 4n_k = n.
    \]
    \item For each Boolean variable $v \in V_k$, Firm~1 has a set $Y_v$ of $2(d_v + 1)$ {\em good robotaxis}. Firm~1 also has a set $Y_B$ of $6m_k + 2n_k$ {\em bad robotaxis}. In total, Firm~1 has
    \[
    |Y| = \sum_{v \in V_k} 2(d_v + 1) + 6m_k + 2n_k = 12m_k + 4n_k = n
    \]
    robotaxis.
    \item For each clause $e = (u,v,w,b_e) \in E_k$, consider the four binary tuples $(a,b,c) \in \{0,1\}^3$ satisfying $a \oplus b \oplus c = b_e \pmod 2$. Firm~2 has one {\em stealer robotaxi} $Z_{e,a,b,c}$ for each clause $e \in E_k$ and satisfying binary tuple $(a,b,c)$. Firm~2 also has a set $Z_D$ of $8m_k + 4n_k$ {\em dummy robotaxis}. The total number of firm 2's robotaxis is
    \[
    |Z| = 4|E_k| + 8m_k + 4n_k = n.
    \]
    \item The distances between riders and robotaxis in $\Gamma(H,k)$ are as follows. For $v \in V_k$, each good robotaxi in $Y_v$ is at distance $1$ from each variable rider in $X_{v,0}$ and $X_{v,1}$ and at distance $3$ from any other rider. For clause $e=(u,v,w,b_e)$ and satisfying binary tuple $(a,b,c)$, stealer robotaxi $Z_{e,a,b,c}$ is at distance $2$ from riders $X_{u,a}, X_{v,b},$ and $X_{w,c}$ and at distance $4$ from any other rider. Bad and dummy robotaxis are at distance $3$ and $4$ from each rider, respectively.
\end{itemize}
The reduction clearly takes time polynomial in the size of instance $H$ and the parameter $k$.

Now, consider a strategy of Firm 1 in game $\Gamma(H,k)$, i.e., a perfect matching from the robotaxis in $Y$ to the riders in $X$. For $v\in V_k$, a rider $x\in X_{v,0}\cup X_{v,1}$ is called {\em protected} if it is matched to a good robotaxi $y\in Y_v$. In this case, the distance between $x$ and $y$ is $1$. The minimum distance of any robotaxi of Firm 2 from rider $x$ is $2$; hence, rider $x$ cannot be stolen by Firm 2. If instead the rider $x$ is matched to a robotaxi $y$ that does not belong to $Y_v$ (i.e., it is a bad robotaxi or a good robotaxi for some variable different than $v$), we call it {\em vulnerable}. The distance between $x$ and $y$ is $3$ and a stealer robotaxi of Firm 2 at distance $2$ can steal the rider.

We say that a strategy $M$ of Firm 1 in game $\Gamma(H,k)$ is {\em canonical} if for every $v\in V_k$, there is $a_v\in \{0,1\}$ so that all the good robotaxis in $Y_v$ are matched to the variable riders in $X_{v,a_v}$ (and the riders in $X_{v,1-a_v}$ are matched to the bad robotaxis in $Y_B$). Notice that there is a one-to-one correspondence between canonical matchings between the riders in $X$ and the robotaxis of Firm 1 in $Y$ and Boolean assignments to the variables in $V_k$. For a canonical matching $M$, we denote by $\sat(M)$ the number of clauses in $E_k$ that are satisfied by the Boolean assignment to the variables of $V_k$ that corresponds to $M$. 

 We will prove the following two lemmas.

\begin{lemma}\label{lem:opt-canonical}
There is an optimal Stackelberg strategy of Firm 1 in game $\Gamma(H,k)$ that is a canonical matching.
\end{lemma}

\begin{proof}
Let $M$ be an optimal Stackelberg strategy for Firm 1. Recall that Firm 1's payoff is given by $n - \nu(M)$, where $\nu(M)$ is the cardinality of a maximum matching in the stealing graph $G(M)$. Thus, an optimal strategy $M$ is a perfect matching between the set of riders $X$ and Firm 1's robotaxis $Y$ that minimizes $\nu(M)$. We will prove the lemma in three steps by transforming $M$ into a canonical matching without increasing $\nu(M)$.

We first show how to transform $M$ to a matching $M'$ in which every good robotaxi in $Y_v$ is matched to a rider in $X_{v,0}\cup X_{v,1}$ for every $v\in V_k$ and, furthermore, $\nu(M')\leq \nu(M)$. Assume that for some $v\in V_k$, the good robotaxi $y\in Y_v$ is matched to a vulnerable rider $x\not\in X_{v,0}\cup X_{v,1}$. Since $|X_{v,0}\cup X_{v,1}|=4(d_v+1)>2(d_v+1)=|Y_v|$, the pigeonhole principle yields that there exists a vulnerable rider $x'\in X_{v,0}\cup X_{v,1}$ which is matched to a robotaxi $y'\not\in Y_v$. Now, consider the matching $\widehat{M}$ obtained from $M$ after replacing the pairs $(x,y)$ and $(x',y')$ with the pairs $(x,y')$ and $(x',y)$. Since rider $x$ is vulnerable in $M$, their incident edges in the stealing graph $G(\widehat{M})$ exist in $G(M)$. Also, rider $x'$ is protected in $G(\widehat{M})$ and thus has degree $0$. Any other rider has the same incident edges in $G(M)$ and $G(\widehat{M})$. Hence, $G(\widehat{M})$ is a subgraph of $G(M)$ and $\nu(\widehat{M})\leq \nu(M)$. Notice that this process decreases the number of good robotaxis matched to vulnerable riders. By repeating it until this number becomes $0$, we obtain the matching $M'$ with the desired property.

Now, a matching satisfying the property that each good robotaxi in $Y_v$ is matched to some (protected) rider in $X_{v,0}\cup X_{v,1}$ for every $v\in V_k$ can be represented by integers $\mu_{v,a}$ denoting the number of vulnerable riders it defines in set $X_{v,a}$ for $v\in V_k$ and $a\in \{0,1\}$. These integers must satisfy the constraint $\mu_{v,0}+\mu_{v,1}=2(d_v+1)$ for every variable $v\in V_k$.

Then, the stealing graph $G(M)$ consists of edges between vulnerable riders and stealer robotaxis. For $v\in V_k$ and $a\in \{0,1\}$, denote by $I_{v,a}$ the set of stealer robotaxis that are at distance $2$ from the riders in $X_{v,a}$. These include the robotaxi $Z_{e,\tau}$, where $e$ is one of the $d_v$ clauses that includes the variable $v$ and $\tau$ is one of the two tuples with binary values for the variables in $e$, using specifically the value of $a$ for $v$. Hence, $|I_{v,a}|=2d_v$. The vulnerable riders in $X_{v,a}$ are connected with the stealer robotaxis in $I_{v,a}$ forming a complete bipartite graph.

By K\H{o}nig's theorem, the maximum matching size $\nu(M)$ in the stealing graph $G(M)$ is equal to the size of the minimum vertex cover. Let $\mathcal{S}$ denote the set of all possible subsets of stealer robotaxis. For any fixed subset $S \in \mathcal{S}$ chosen to be in the vertex cover, the remaining edges of $G(M)$ must be covered using riders. 

Recall that the stealers in $I_{v,a}$ form a complete bipartite subgraph with the vulnerable riders in $X_{v,a}$. If $I_{v,a} \not\subseteq S$, there is at least one stealer robotaxi in $I_{v,a}$ not in the cover. Because the subgraph is complete bipartite, all $\mu_{v,a}$ vulnerable riders in $X_{v,a}$ must be included in the vertex cover to cover these edges. Conversely, if $I_{v,a} \subseteq S$, no riders from $X_{v,a}$ are required. Thus, the minimum vertex cover size for a given matching configuration $\mu$ is:
\[
\nu(M) = \min_{S \in \mathcal{S}} \left( |S| + \sum_{v \in V_k} \sum_{a \in \{0,1\}} \mu_{v,a} \cdot \one\{I_{v,a} \not\subseteq S\} \right)
\]
Firm 1's objective is to choose a valid matching $M$ (which determines $\mu$) to minimize $\nu(M)$. This is equivalent to finding the global minimum of the joint function $f(\mu, S) = |S| + \sum_{v,a} \mu_{v,a} \one\{I_{v,a} \not\subseteq S\}$ over all valid rider configurations $\mu$ and stealer subsets $S$.

By the properties of joint minimization, we have $\min_\mu \min_S f(\mu,S) = \min_S \min_\mu f(\mu,S)$. Let $S^*$ be the optimal stealer subset for the global minimum. For this fixed $S^*$, Firm 1's optimal configuration $\mu^*$ must minimize the inner sum:
\[
\sum_{v \in V_k} \left( \mu_{v,0} \cdot \one\{I_{v,0} \not\subseteq S^*\} + \mu_{v,1} \cdot \one\{I_{v,1} \not\subseteq S^*\} \right)
\]
subject to $\mu_{v,0} + \mu_{v,1} = 2d_v + 2$ for each $v$. Because this objective is strictly linear with respect to the variables $\mu_{v,0}$ and $\mu_{v,1}$, the minimum on the simplex is always achieved at an extreme point (the boundary). Specifically, for each $v$:
\begin{itemize}
    \item If $\one\{I_{v,0} \not\subseteq S^*\} < \one\{I_{v,1} \not\subseteq S^*\}$, the minimum is achieved by $\mu_{v,0} = 2d_v+2$ and $\mu_{v,1} = 0$.
    \item If $\one\{I_{v,0} \not\subseteq S^*\} \geq \one\{I_{v,1} \not\subseteq S^*\}$, the minimum is achieved by $\mu_{v,0} = 0$ and $\mu_{v,1} = 2d_v+2$.
\end{itemize}
In either case, the optimal solution places all $2d_v + 2$ vulnerable riders on one side and $0$ on the other. This boundary configuration exactly matches the definition of a canonical matching. Therefore, there always exists an optimal Stackelberg strategy $M^*$ that is canonical, completing the proof.
\end{proof}

\begin{lemma}\label{lem:payoff}
For any canonical matching $M$, the Stackelberg payoff of Firm 1 is $8m_k+4n_k+\sat(M)$. 
\end{lemma}

\begin{proof}
We will first show that given any canonical matching $M$ between the riders in $X$ and the robotaxis of Firm 1 in $Y$, the minimum vertex cover of the stealing graph consists entirely of stealer robotaxis and contains no riders. This establishes that $\nu(M)$ is equal to the number of stealer robotaxis with positive degree in the stealing graph $G(M)$. Finally, we will prove that this number is equal to $4m_k-\sat(M)$, which yields the claimed Stackelberg payoff.

Consider a canonical matching $M$ and let $S$ be a minimum vertex cover of the stealing graph $G(M)$. For the sake of contradiction, assume that $S$ contains a rider $x\in X_{v,a}$ for some variable $v\in V_k$ and $a\in \{0,1\}$. By the minimality of $S$, $x$ must be a vulnerable rider; otherwise, $x$ would have degree $0$ in the stealing graph $G(M)$ and could be simply removed. Furthermore, there must exist a stealer robotaxi $z$ adjacent to $x$ in $G(M)$ that does not belong to $S$. If this were not the case, $x$ could be removed from $S$ without leaving any of its incident edges uncovered. Because all the vulnerable riders in $X_{v,a}$ share the exact same neighborhood in $G(M)$ (including the robotaxi $z$), none of these edges can be covered by $z$, meaning that all the $2d_v+2$ vulnerable riders in $X_{v,a}$ must be included in $S$. 

However, we can satisfy the covering requirements for the edges incident to the riders in $X_{v,a}$ in $G(M)$ by replacing these $2d_v+2$ riders in $S$ with their neighboring stealer robotaxis. These neighbors are stealer robotaxis of the form $Z_{e,\tau}$, where $e$ is one of the $d_v$ clauses containing variable $v$, and $\tau$ is one of the two satisfying tuples for $e$ that specifically assigns the value $a$ to variable $v$. Hence, we can replace the $2d_v+2$ riders of $X_{v,a}$ in $S$ with their $2d_v$ neighboring stealer robotaxis and have a valid vertex cover of the stealing graph of strictly smaller size; this contradicts the minimality of $S$.

Therefore, the minimum vertex cover in the stealing graph $G(M)$ consists exclusively of stealer robotaxis. Its size $\nu(M)$ is simply the number of stealer robotaxis with a positive degree in $G(M)$. We now compute this number. 

Let $a=(a_v)_{v\in V_k}$ be the Boolean assignment corresponding to the canonical matching $M$. Recall that for every $v\in V_k$, the riders in $X_{v,a_v}$ are protected, and the riders in $X_{v,1-a_v}$ are vulnerable. For a clause $e=(u,v,w)$ that is satisfied by assignment $a$, the four corresponding stealer robotaxis are $Z_{e,a_u,a_v,a_w}$, $Z_{e,a_u,1-a_v,1-a_w}$, $Z_{e,1-a_u,a_v,1-a_w}$, and $Z_{e,1-a_u,1-a_v,a_w}$. Among these, $Z_{e,a_u,a_v,a_w}$ has only protected riders (i.e., those in $X_{u,a_u}\cup X_{v,a_v} \cup X_{w,a_w}$) at distance smaller than $4$ and, thus, has degree $0$ in the stealing graph $G(M)$. The other three stealer robotaxis have some vulnerable rider at distance $2$ (belonging to either $X_{u,1-a_u}$ or $X_{v,1-a_v}$ or $X_{w,1-a_w}$), giving them positive degree in $G(M)$. For a clause $e=(u,v,w)$ that is not satisfied by assignment $a$, the corresponding stealer robotaxis are $Z_{e,1-a_u,a_v,a_w}$, $Z_{e,a_u,1-a_v,a_w}$, $Z_{e,a_u,a_v,1-a_w}$, and $Z_{e,1-a_u,1-a_v,1-a_w}$. All these stealer robotaxis have some vulnerable riders at a distance smaller than $4$ and thus have a positive degree in the stealing graph, too. 

Hence, the total number of stealer robotaxis with a positive degree in $G(M)$ is $4m_k-\sat(M)=\nu(M)$. We conclude that the Stackelberg payoff of Firm 1 when using the canonical matching $M$ as strategy is
\[
n - \nu(M) = (12m_k + 4n_k) - (4m_k - \sat(M)) = 8m_k + 4n_k + \sat(M). \qedhere
\]
\end{proof}

Next, we would like to get a relation to the optimal number of satisfying clauses in the \textsc{Max-3-Lin(2)} instance $H$ and the optimal Stackelberg payoff of Firm 1 in the game $\Gamma(H,k)$. This will be possible via the next lemma.

\begin{lemma}\label{lem:blow-up}
For every \textsc{Max-3-Lin(2)} instance $H$ and positive integer $k$, it holds
\[
\opt(H^k) = k^3 \cdot \opt(H).
\]
\end{lemma}

\begin{proof}
Setting all copies in each cloud $C_v$ to the value of $v$ in an optimal assignment of $H$ satisfies all $k^3$ copies of every satisfied clause, giving $\opt(H^k) \geq k^3 \cdot \opt(H)$. 
Conversely, in the optimal assignment of $H^k$, selecting a variable copy from each cloud $C_v$ uniformly at random defines an assignment for $H$ whose expected number of satisfied clauses is exactly $\frac{1}{k^3}\opt(H^k) \leq \opt(H)$, which implies $\opt(H^k) \leq k^3 \cdot \opt(H)$.
\end{proof}

Our next statement gives the main property of our reduction.

\begin{lemma}\label{lem:reduction}
    Given an instance $H=(V,E)$ of \textsc{Max-3-Lin(2)} and an integer $k$, the optimal Stackelberg payoff of Firm 1 is $8k^3m+4kN+k^3 \opt(H)$.
\end{lemma}

\begin{proof}
Using Lemmas~\ref{lem:opt-canonical} and~\ref{lem:payoff}, we have that for every instance $H=(V,E)$ of \textsc{Max-3-Lin(2)} and positive integer $k$, the optimal Stackelberg payoff of Firm 1 is $8m_k+4n_k+\opt(H^k)=8k^3 m + 4k N +k^3 \opt(H)$, as desired. The equality follows from the definitions of $m_k$, $n_k$ and Lemma~\ref{lem:blow-up}.
\end{proof}

We are now ready to complete the proof of Theorem~\ref{thm:n_over_24}. Given a constant $\delta\in (0,1/24)$ and an instance $H=(V,E)$ of \textsc{Max-3-Lin(2)} with $N$ variables and $m$ clauses, we construct the $n$-robotaxi assignment game $\Gamma(H,k)$ for $k=\left\lceil \sqrt{\frac{N}{36\delta m}}\right\rceil$ and set $\varepsilon=3\delta$. Since $\delta$ is a positive constant, the parameter $k$ is polynomial in the size of instance $H$ and the construction takes polynomial time. Using H\aa stad's inapproximability result for \textsc{Max-3-Lin(2)} (Theorem~\ref{thm:hastad}) and Lemma~\ref{lem:reduction}, we get that it is NP-hard to distinguish between the following two cases for game $\Gamma(H,k)$:
\begin{itemize}
    \item The optimal Stackelberg payoff of Firm 1 is at least $9k^3m+4kN-\varepsilon k^3 m$.
    \item The optimal Stackelberg payoff of Firm 1 is at most $8.5k^3m+4kN+\varepsilon k^3 m$.
\end{itemize}
Hence, it is NP-hard to approximate the optimal Stackelberg payoff of Firm 1 within an additive error better than the difference of these quantities, i.e., 
\begin{align*}
\frac{k^3m}{2}-2\varepsilon k^3m &= \frac{n-4kN}{24}-2\varepsilon\, k^3 m = \frac{n}{24}-\left(\frac{N}{6k^2m}+2\varepsilon \right) k^3m\\
&\geq \left(\frac{1}{24}-\frac{N}{72k^2m}-\frac{\varepsilon}{6}\right) n \geq \left(\frac{1}{24}-\delta\right)n,
\end{align*}
as desired. The first two equalities follow by the definitions of $n$, $n_k$, and $m_k$ (i.e., $n=12m_k+4n_k$, $n_k=k N$, and $m_k=k^3 m$), the first inequality by the fact $k^3m\leq n/12$, and the second one by the definition of $\varepsilon$ and $k$.

\section{A dynamic programming algorithm for robotaxis and riders on the line}\label{sec:1D}

In this section, we present a polynomial-time dynamic programming algorithm to compute the optimal Stackelberg payoff for Firm 1 when robotaxis and riders are located on a line. Recall that the objective of Firm 1 is to compute a matching $M$ from its robotaxis in set $Y$ to the riders in set $X$ so that the quantity $n-\nu(G(M))$ is maximized, where $\nu(G(M))$ is the size of the minimum vertex cover in the stealing graph defined by the matching $M$. Assume that a minimum vertex cover in the stealing graph consists of the set of riders $S\subseteq X$ and the set of robotaxis $T\subseteq Z$ of Firm 2 and define the complementary sets $U=X\setminus S$ and $W=Z\setminus T$. Then, the optimal Stackelberg payoff of Firm 1 can be expressed as
\[n-\nu(G(M)) = n-(|S|+|T|)=|U|+|W|-n.\]
By the vertex cover definition, we have that 
\[\mbox{$\forall x\in U, \forall z\in W: (x,z)\not\in G(M)$},\]
and by the definition of the edge set of the stealing graph $G(M)$, the vertex cover definition is equivalent to
\begin{align}\label{eq:w-constraint}
\mbox{$\forall x\in U, \forall z\in W: d(x,M(x))\leq d(x,z)$}.
\end{align}
This property allows us to compute the optimal Stackelberg payoff in the following rather counter-intuitive way. Assuming that we have guessed a set $W$ of robotaxis of Firm 2, we would like to compute a maximum-cardinality set of riders $U$ so that there exists a matching $M$ as a strategy of Firm 1 that yields a stealing graph $G(M)$ with $(X\setminus U)\cup (Z\setminus W)$ as vertex cover. We will argue about the cardinality of $U$ via the matching $M$, by arguing about the number of edges in $M$ that are incident to riders of $U$. Optimizing over the different possible selections of $W$ will then give us the optimal Stackelberg payoff. 

For a location $p$, we denote by $X_p$, $Y_p$, and $Z_p$ the sets of riders, robotaxis of Firm 1, and robotaxis of Firm 2 that are located at point $p$, respectively. For a set $W\subseteq Z$ of robotaxis of Firm 2, we denote by $w_1, w_2, ..., w_k$ their locations so that $w_t<w_{t+1}$ for $t\in [k-1]$ and use the term {\em boundary} alternatively to location for each of them (a boundary $w_t$ can host more than one robotaxi in $W$). For $t\in [k]$, let $z_t$ be a robotaxi of Firm 2 in $W$ located at the boundary $w_t$. We also set $w_0=-\infty$ and $w_{k+1}=+\infty$. We partition the riders in $U$ into those located at the boundaries $w_1, w_2, ..., w_k$ and those located elsewhere. \begin{itemize}
    \item First, consider a boundary $w_t$ for some $t\in [k]$ and notice that a rider $x\in X_{w_t}$ can belong to $U$ only if its matched robotaxi $M(x)$ of Firm 1 is also located at $w_t$. Otherwise, the constraint (\ref{eq:w-constraint}) would be violated. As we are interested in maximizing the cardinality of set $U$, we can include as many riders of $X_{w_t}$ as we can match with robotaxis in $Y_{w_t}$, i.e., $\min\{|X_{w_t}|,|Y_{w_t}|\}$ riders, in $U$. Among the $\max\{0,|Y_{w_t}|-|X_{w_t}|\}$ {\em surplus} robotaxis in $Y_{w_t}$ that are not matched to riders in $X_{w_t}\cap U$, let $L_{w_t}$ and $R_{w_t}$ denote the number of {\em serving-left} and {\em serving-right} surplus robotaxis of $Y_{w_t}$ that are matched under $M$ to riders of $U$ located strictly at the left and strictly at the right of the boundary $w_t$, respectively. Clearly, $L_{w_t}+R_{w_t}\leq \max\{0,|Y_{w_t}|-|X_{w_t}|\}$.

    \item Next, consider a rider $x$ located in the open interval $(w_t,w_{t+1})$ for some $t\in [k]\cup \{0\}$. Rider $x$ belongs to $U$ only if they do not violate the constraint (\ref{eq:w-constraint}). This means that their matched robotaxi $M(x)$ of Firm 1 cannot be located strictly at the left of boundary $w_t$ if $t\in [k]$ or strictly at the right of boundary $w_{t+1}$ if $t\in [k-1]\cup \{0\}$; otherwise, we would have $d(x,M(x))>d(x,z_t)$ or $d(x,M(x))>d(x,z_{t+1})$, violating the constraint (\ref{eq:w-constraint}). Hence, rider $x\in U$ can be matched either to some of the $R_{w_t}$ serving-right surplus robotaxis of Firm 1 located at $w_t$, or to some of the $L_{w_{t+1}}$ serving-left surplus robotaxis of Firm 1 located at $w_{t+1}$, or to some robotaxi of Firm 1 located in the open interval $(w_t,w_{t+1})$, so that the constraint (\ref{eq:w-constraint}) is satisfied for the robotaxis $z_t$ and $z_{t+1}$ of Firm 2 in set $W$. We can compute the maximum number of riders located in $(w_t,w_{t+1})$ that also belong to $U$ by computing a maximum-cardinality matching in the bipartite graph that has all riders in $X$ located in the open interval $(w_t,w_{t+1})$ as vertices on one side, the $R_{w_t}$ serving-right surplus robotaxis of Firm 1 located at the boundary $w_t$, the $L_{w_{t+1}}$ serving-left surplus robotaxis of Firm 1 located at the boundary $w_{t+1}$, and all robotaxis of Firm 1 located in the open interval $(w_t,w_{t+1})$ as vertices on the other side, so that the rider-robotaxi edge $(x,y)$ exists in the graph if and only if $d(x,y)\leq \min\{d(x,w_t),d(x,w_{t+1})\}$. 
\end{itemize}

Notice that it is the location of a robotaxi of Firm 2 in $W$ and not the robotaxi itself that affects the riders that can be included in set $U$. Thus, if $p$ is the location of a robotaxi of Firm 2 that can also be included in $W$, then all robotaxis of Firm 2 that are located at $p$ also belong in $W$. Furthermore, given the distinct locations of the robotaxis of Firm 2 in $W$ (i.e., the boundaries), we can maximize the number of riders that can be included in $U$ by considering the riders in each boundary and each open interval between consecutive boundaries separately. This is the idea behind the following dynamic programming algorithm, which computes the optimal Stackelberg payoff of Firm 1.

Given a robotaxi assignment game, let $p_1, p_2, ..., p_m$ be the distinct locations of all robotaxis of Firm 2 in $Z$, in increasing order by their coordinates. We also set $p_0=-\infty$ and $p_{m+1}=+\infty$ and use the convention $X_{-\infty}=Y_{-\infty}=X_{+\infty}=Y_{+\infty}=\emptyset$. For integers $i,j,r,\ell$ with $0\leq i<j\leq m+1$, $0\leq r\leq \max\{0,|Y_{p_i}|-|X_{p_i}|\}$, and $0\leq \ell \leq \max\{0,|Y_{p_j}|-|X_{p_j}|\}$, we use $U(i,j, r, \ell)$ to denote the maximum number of riders in $U$ located in the open interval $(p_i,p_j)$ when $p_i$ and $p_j$ are consecutive boundaries (i.e., the robotaxis of Firm 2 located at $p_i$ and $p_j$, and no robotaxi of Firm 2 in between, have been selected to be included in set $W$), $r$ serving-right surplus robotaxis of Firm 1 are located at location $p_i$, and $\ell$ serving-left surplus robotaxis of Firm 1 are located at location $p_j$.

For integers $j\in \{0, 1, ..., m+1\}$ and $c\in \{0, 1, ..., \max\{0,|Y_{p_j}|-|X_{p_j}|\}\}$, we use $\DP(j,c)$ to denote the maximum contribution of riders and robotaxis of Firm 2 in the quantity $|U|+|W|$ located in $(-\infty,p_j]$, assuming that we place the robotaxis in $p_j$ in $W$ and reserve exactly $c$ serving-right surplus robotaxis at location $p_j$. As initial condition, we set $\DP(0,0)=0$. The other entries $\DP(j,c)$ are computed recursively as
\begin{align*}
\DP(j,c) &= |Z_{p_j}|+\min\{|X_{p_j}|,|Y_{p_j}|\}+\max_{0\leq i<j}\max_{R_{p_i}}\max_{L_{p_j}}{\left\{\DP(i,R_{p_i})+U(i,j,R_{p_i},L_{p_j})\right\}},
\end{align*}
where the range of $R_{p_i}$ and $L_{p_j}$ in the two $\max$'s are $0\leq R_{p_i}\leq \max\{0,|Y_{p_i}|-|X_{p_i}|\}$, and $0\leq L_{p_j}\leq \max\{0,|Y_{p_j}|-|X_{p_j}|-c\}$, respectively. The optimal Stackelberg payoff is then given by $\DP(m+1,0)-n$.

We now argue that the dynamic programming algorithm above can be implemented to run in time $O(n^3)$. We assume that we have access to the quantities $|X_p|$, $|Y_p|$, and $|Z_p|$ for every location $p$ holding some rider or some robotaxi. Also, each of the quantities $U(i,j,r,\ell)$ is obtained by computing a maximum-cardinality matching in a {\em convex} bipartite graph (i.e., a bipartite graph whose vertices on one side can be ordered such that the neighborhood of each vertex in the opposite partition forms a contiguous interval). This can be done in linear time in terms of the number of vertices in the graph (i.e., at most $O(n)$) using the algorithm of Gabow and Tarjan~\cite{GT85}. Considering the different values the parameters $i$, $j$, $r$, and $\ell$ can take, the number of such computations is at most
\begin{align*}
    &\sum_{i=0}^m\sum_{j=i+1}^{m+1}(1+\max\{0,|Y_{p_i}|-|X_{p_i}|\})(1+\max\{0,|Y_{p_j}|-|X_{p_j}|\}) \leq \sum_{i=0}^m\sum_{j=i+1}^{m+1}(1+|Y_{p_i}|)(1+|Y_{p_j}|)\\
    &\leq \frac{1}{2}\cdot \left(\sum_{i=0}^{m+1}{(1+|Y_{p_i}|)}\right)^2=\frac{1}{2}\cdot \left(m+2+\sum_{i=0}^{m+1}{|Y_{p_i}|}\right)^2 \leq 2 (n+1)^2,
\end{align*}
since both $m$ and $\sum_{i=0}^{m+1}{|Y_{p_i}|}$ are upper-bounded by $n$.
Hence, the precomputation of the $U(\cdot)$ values takes $O(n^3)$ time.

Once the values $U(\cdot)$ have been computed, we can bound the time to evaluate all the $\DP(j,c)$ entries by considering the possible values the parameters $j$ and $c$ can take and the number of combinations of parameters $i$, $R_{p_i}$, and $L_{p_j}$ for each of them. So, the time to evaluate all the $\DP(j,c)$ entries is 
\begin{align*}
&\sum_{j=1}^{m+1}\sum_{c=0}^{\max\{0,|Y_{p_j}|-|X_{p_j}|\}}\sum_{i=0}^{j-1}(1+\max\{0,|Y_{p_i}|-|X_{p_i}|\})(1+\max\{0,|Y_{p_j}|-|X_{p_j}|\}-c)\\
&\leq \sum_{j=1}^{m+1}\sum_{i=0}^{j-1}(1+|Y_{p_i}|)(1+|Y_{p_j}|)^2 \leq \frac{1}{3}\cdot \left(\sum_{j=0}^{m+1}(1+|Y_{p_j}|)\right)^3 \leq \frac{8}{3}\cdot (n+1)^3.
\end{align*}

Hence, the running time of the algorithm is at most $O(n^3)$. The next statement summarizes the discussion in this section.

\begin{theorem}
    On input a robotaxi assignment game with $n$ riders and two firms with $n$ robotaxis each, all located on a line, the algorithm described above runs in time $O(n^3)$ and computes the optimal Stackelberg payoff of Firm~1. 
\end{theorem}

\section{Experimental results}\label{sec:exp}
We evaluate our algorithms using the same experimental setup introduced in the robotaxi placement paper~\cite{CKRSS26}. The metric space consists of 263 points corresponding to the centroids of the New York City taxizones, as defined by the NYC Taxi and Limousine Commission (TLC).\footnote{\url{https://www.nyc.gov/site/tlc/}} The demand distribution over these locations is derived from the 2023 High Volume For-Hire Vehicle Trip Dataset provided by TLC.\footnote{\url{https://data.cityofnewyork.us/Transportation/2023-High-Volume-FHV-Trip-Data/u253-aew4/about_data}} Specifically, we utilize the empirical distribution of pickup locations between 8:00 AM and 9:00 AM during weekdays (March 20--24, 2023).\footnote{All relevant data as well as the code can be found as supplementary material.}

For our Stackelberg game experiments, we evaluate instances with $n \in \{50, 100, 200\}$ riders and robotaxis per firm. For each fleet size $n$, we generated 100 independent instances. In each instance, the locations of the leader's fleet (Firm 1), the follower's fleet (Firm 2), and the $n$ riders are sampled independently according to the probability distribution on pickup locations.

We have implemented the following algorithms to compute the matching strategy for the leader:

\begin{table*}[t]
    \centering
    \begin{tabular}{r c c c c c}
        \hline
        $n$ & \textbf{Min-Cost} & \textbf{Greedy} & \textbf{DAG}  & \textbf{LP+BvN} & \textbf{LP Bound} \\
        \hline
        50  & $19.99 \pm 0.63$ & $22.47 \pm 0.58$ & $23.04 \pm 0.58$ & $24.41 \pm 0.59$ & $26.79 \pm 0.54$ \\
        100 & $43.92 \pm 0.87$ & $49.50 \pm 0.77$ & $50.57 \pm 0.76$ & $52.81 \pm 0.78$ & $58.08 \pm 0.66$ \\
        200 & $102.23 \pm 1.32$ & $112.21 \pm 1.08$ & $113.38 \pm 1.04$ & $114.81 \pm 1.06$ & $125.98 \pm 0.96$ \\
        \hline
    \end{tabular}
    \caption{Comparison of the leader's Stackelberg payoff across different matching algorithms for fleet size $n \in \{50, 100, 200\}$.}
    \label{tab:payoff}
\end{table*}

\begin{itemize}
    \item \textbf{Min-Cost:} This algorithm naively solves the standard minimum-cost bipartite matching between the leader's robotaxis and the riders, completely ignoring the follower's presence.
    \item \textbf{Greedy:} This algorithm iteratively selects the minimum-cost edge that is not adjacent to edges included in the matching between the leader's fleet and the riders earlier.
    \item \textbf{Defense-Aware Greedy (DAG):} A refined greedy approach that actively penalizes assignments where the follower's closest robotaxi can easily intercept the rider. It runs the greedy algorithm on the modified cost matrix $\widehat{d}$ with $\widehat{d}(x,y)=d(x,y)+\max\{0,d(x,y)-\min_{z\in Z}{d(x,z)}\}$, for every rider $x\in X$ and robotaxi $y\in Y$ of Firm 1. 
    \item \textbf{LP+BvN:} This is our LP-based approximation algorithm, which computes the fractional solution to the LP relaxation of the problem, applies a Birkhoff-von Neumann decomposition of the fractional solution into a polynomial number of integral matchings, and selects the one that yields the highest Stackelberg payoff.
    \item \textbf{LP Bound:} For comparison, we compute and report the objective value of the LP relaxation.
\end{itemize}

All algorithms, data processing pipelines, and experimental evaluations were implemented in Python 3. For computing standard minimum-cost matchings, we utilized the Python Optimal Transport (POT) library, specifically the highly optimized \texttt{ot.lp.emd} function. The linear programs were solved using SciPy's HiGHS Dual Simplex solver, while the follower's deterministic best-response stealing phase was evaluated using SciPy's sparse maximum bipartite matching implementation.

The payoff performance results are depicted in Table~\ref{tab:payoff}. Each number indicates the mean Stackelberg payoff for the leader, accompanied by its 95\% confidence interval computed across the 100 independent runs. In all cases, the Min-Cost algorithm yields the worst payoff, demonstrating that naively minimizing travel distance without anticipating the follower's predatory strategy leads to severe losses. The Greedy algorithm provides a surprisingly strong baseline, significantly outperforming Min-Cost. Incorporating defensive penalties proves beneficial; the DAG algorithm consistently improves the retained rider count over the standard Greedy approach. Ultimately, our LP-based algorithm yields the best empirical payoffs and achieves the highest retention across all fleet sizes, closely tracking the theoretical LP bound.

While the LP-based algorithm achieves the highest payoffs, this performance comes with a steep computational trade-off. We tracked the execution timings across the 300 instances and found that the running times of both the Greedy and DAG algorithms are substantially better than the LP-based algorithm. For smaller instances ($n=50$), the greedy variants execute in roughly 2 milliseconds on average, while the complete LP+BvN pipeline takes nearly 200 milliseconds per instance. As the fleet size scales to $n=200$, this gap widens drastically. Both Greedy and DAG remain extremely fast, taking approximately 15 milliseconds per instance. In contrast, solving the LP relaxation becomes a significant bottleneck, pushing the execution time of the LP-based algorithm to nearly 45 seconds per instance.

Ultimately, our experiments convey a clear, twofold message. First, in practice, the performance of the LP+BvN algorithm is considerably better than what its theoretical $n/2$-additive bound dictates, consistently securing payoffs that are highly competitive and close to the objective value of the linear program. Second, there are simpler and drastically faster algorithms---namely, the Greedy and DAG approaches---that perform exceptionally well empirically. Even though these simpler heuristics do not provide any non-trivial theoretical guarantees, their strong empirical performance and negligible computational overhead make them highly attractive and practical alternatives for large-scale deployments.

\section{Conclusion and open problems}\label{sec:open}

In this paper, we introduced and analyzed the Stackelberg robotaxi assignment game, modeling the competitive spatial routing dynamics between two centrally controlled autonomous mobility fleets. By formulating the follower's best response via the stealing graph and K\H{o}nig's theorem, we bypassed the min-max bottleneck to design an exact ILP, which we relaxed and rounded using the Birkhoff--von Neumann decomposition to achieve a deterministic Stackelberg payoff within an $n/2$-additive error of the optimum. Complementing this with matching integrality gaps, we established that multiplicative approximations are unattainable (via NP-completeness of deciding strictly positive payoffs) and that an $\Omega(n)$ additive approximation error is structurally inevitable on general metrics. On the positive side, we demonstrated tractability in restricted spatial topologies, developing an $O(n^3)$ exact dynamic program on the Euclidean line. Finally, our empirical evaluation on real-world mobility data revealed that our LP-based algorithm significantly outperforms its theoretical worst-case bound, while simple defense-aware heuristics provide highly competitive empirical payoffs with exceptional computational efficiency.

Our work opens several intriguing avenues for future research. From an approximability standpoint, an immediate open question is to tighten the quantitative gap between our $n/2$ upper bound and the $(1/24 - \delta)\,n$ hardness lower bound for the optimal additive approximation guarantee. Another critical direction is exploring geometric metrics in higher dimensions: the computational complexity on the 2D Euclidean plane or in shortest-path metrics in planar graphs---the most operationally relevant settings for urban road networks---remains open. Does the problem of computing optimal Stackelberg payoff admit a polynomial-time approximation scheme (PTAS), or is it also APX-hard in these cases? Finally, natural game-theoretic generalizations warrant investigation, such as asymmetric fleet capacities and oligopolies with three or more competing firms.

\bibliographystyle{plain}
\bibliography{stackelberg}

@article{korzhyk2011stackelberg,
  title={Stackelberg vs. Nash in security games: An extended investigation of interchangeability, equivalence, and uniqueness},
  author={Korzhyk, Dmytro and Yin, Zhengyu and Kiekintveld, Christopher and Conitzer, Vincent and Tambe, Milind},
  journal={Journal of Artificial Intelligence Research},
  volume={41},
  pages={297--327},
  year={2011}
}

@article{CKRSS26,
  author       = {Ioannis Caragiannis and
                  Kostas Kollias and
                  Mohammad Roghani and
                  Aaron Schild and
                  Ali Kemal Sinop},
  title        = {The Robotaxi Placement Problem: Minimizing Expected {ETA} for Stochastic
                  Demand},
  journal      = {CoRR},
  volume       = {abs/2605.15745},
  year         = {2026},
  eprinttype   = {arXiv},
  eprint       = {2605.15745}
}

@incollection{dantzig1951equivalence,
  author    = {George B. Dantzig},
  title     = {A Proof of the Equivalence of the Programming Problem and the Game Problem},
  booktitle = {Activity Analysis of Production and Allocation},
  editor    = {Tjalling C. Koopmans},
  pages     = {330--335},
  publisher = {John Wiley \& Sons},
  year      = {1951}
}

@article{vonstengel2010leadership,
  title={Leadership games with convex strategy sets},
  author={von Stengel, Bernhard and Zamir, Shmuel},
  journal={Games and Economic Behavior},
  volume={69},
  number={2},
  pages={446--457},
  year={2010}
}

@inproceedings{kiekintveld2009computing,
  author    = {Christopher Kiekintveld and Jason Tsai and Manish Jain and Fernando Ord{\'o}{\~n}ez and Milind Tambe},
  title     = {Computing optimal randomized resource allocations for massive security games},
  booktitle = {Proceedings of the 8th International Conference on Autonomous Agents and Multiagent Systems (AAMAS)},
  pages     = {689--696},
  year      = {2009}
}

@article{adler2013equivalence,
  author    = {Ilan Adler},
  title     = {The equivalence of linear programs and zero-sum games},
  journal   = {International Journal of Game Theory},
  volume    = {42},
  number    = {1},
  pages     = {165--177},
  year      = {2013}
}

@article{bimpikis2019spatial,
  title={Spatial pricing in ride-sharing networks},
  author={Bimpikis, Kostas and Candogan, Ozan and Saban, Daniela},
  journal={Operations Research},
  volume={67},
  number={3},
  pages={744--769},
  year={2019}
}

@article{ke2020pricing,
  title={Pricing and equilibrium in on-demand ride-pooling markets},
  author={Ke, Jintao and Yang, Hai and Li, Xinwei and Wang, Hai and Ye, Jieping},
  journal={Transportation Research Part B: Methodological},
  volume={139},
  pages={411--431},
  year={2020},
  publisher={Elsevier}
}

@article{bryan2019theory,
  title={A theory of multihoming in rideshare competition},
  author={Bryan, Kevin A. and Gans, Joshua S.},
  journal={Journal of Economics \& Management Strategy},
  volume={28},
  number={1},
  pages={89--96},
  year={2019}
}

@article{siddiq2022ride,
  title={Ride-hailing platforms: Competition and autonomous vehicles},
  author={Siddiq, Auyon and Taylor, Terry A.},
  journal={Manufacturing \& Service Operations Management},
  volume={24},
  number={3},
  pages={1511--1528},
  year={2022}
}

@article{ahn2004competitive,
  title={Competitive facility location: the Voronoi game},
  author={Ahn, Hee-Kap and Cheng, Siu-Wing and Cheong, Otfried and Golin, Mordecai and van Oostrum, Ren{\'e}},
  journal={Theoretical Computer Science},
  volume={310},
  number={1-3},
  pages={457--467},
  year={2004}
}

@article{cheong2004one,
  title={The one-round Voronoi game},
  author={Cheong, Otfried and Har-Peled, Sariel and Linial, Nathan and Matou{\v{s}}ek, Ji{\v{r}}{\'\i}},
  journal={Discrete \& Computational Geometry},
  volume={31},
  number={1},
  pages={125--138},
  year={2004}
}

@article{hotelling1929stability,
  title={Stability in competition},
  author={Hotelling, Harold},
  journal={The Economic Journal},
  volume={39},
  number={153},
  pages={41--57},
  year={1929}
}

@inproceedings{conitzer2006computing,
  title={Computing the optimal strategy to commit to},
  author={Conitzer, Vincent and Sandholm, Tuomas},
  booktitle={Proceedings of the 7th ACM conference on Electronic commerce (EC)},
  pages={82--90},
  year={2006}
}

@inproceedings{korzhyk2010complexity,
  title={Complexity of computing optimal {S}tackelberg strategies in security resource allocation games},
  author={Korzhyk, Dmytro and Conitzer, Vincent and Parr, Ronald},
  booktitle={Proceedings of the 24th AAAI Conference on Artificial Intelligence},
  pages={805--810},
  year={2010}
}

@article{briest2010stackelberg,
  title={Stackelberg network pricing games},
  author={Briest, Patrick and Hoefer, Martin and Krysta, Piotr},
  journal={Algorithmica},
  volume={62},
  number={3-4},
  pages={733--753},
  year={2012}
}

@article{bohnlein2023stackelberg,
  title={Stackelberg packing games},
  author={B{\"o}hnlein, Toni and Schaudt, Oliver and Schauer, Joachim},
  journal={Theoretical Computer Science},
  volume={943},
  pages={16--35},
  year={2023}
}

@article{GT85,
  title={A linear-time algorithm for a special case of disjoint set union},
  author={Gabow, Harold N. and Tarjan, Robert E.},
  journal={Journal of Computer and System Sciences},
  volume={30},
  number={2},
  pages={209--221},
  year={1985}
}

@article{hastad2001some,
  title={Some optimal inapproximability results},
  author={H{\aa}stad, Johan},
  journal={Journal of the ACM},
  volume={48},
  number={4},
  pages={798--859},
  year={2001},
  publisher={ACM New York, NY, USA}
}

@incollection{karp1972reducibility,
  title={Reducibility among combinatorial problems},
  author={Karp, Richard M.},
  booktitle={Complexity of Computer Computations},
  editor={Miller, Raymond E. and Thatcher, James W.},
  pages={85--103},
  year={1972},
  publisher={Plenum Press},
  address={New York, NY}
}

@article{hall1935representatives,
  title={On representatives of subsets},
  author={Hall, Philip},
  journal={Journal of the London Mathematical Society},
  volume={10},
  number={1},
  pages={26--30},
  year={1935},
  publisher={Wiley Online Library}
}

@book{cormen2009introduction,
  title={Introduction to Algorithms},
  author={Cormen, Thomas H. and Leiserson, Charles E. and Rivest, Ronald L. and Stein, Clifford},
  edition={3rd},
  year={2009},
  publisher={MIT Press}
}

@book{schrijver1986theory,
  title={Theory of Linear and Integer Programming},
  author={Schrijver, Alexander},
  year={1986},
  publisher={John Wiley \& Sons}
}

@article{birkhoff1946tres,
  title={Tres observaciones sobre el algebra lineal},
  author={Birkhoff, Garrett},
  journal={Univ. Nac. Tucum{\'a}n. Rev. Ser. A},
  volume={5},
  pages={147--151},
  year={1946}
}

@article{konig1931,
  title={Gr{\'a}fok {\'e}s m{\'a}trixok},
  author={K{\H{o}}nig, D{\'e}nes},
  journal={Matematikai {\'e}s Fizikai Lapok},
  volume={38},
  pages={116--119},
  year={1931}
}

@book{von1934marktform,
  title={Marktform und Gleichgewicht},
  author={von Stackelberg, Heinrich},
  year={1934},
  publisher={Springer},
  address={Vienna}
}

@article{von1953certain,
  title={A certain zero-sum two-person game equivalent to the optimal assignment problem},
  author={von Neumann, John},
  journal={Contributions to the Theory of Games},
  volume={2},
  pages={5--12},
  year={1953},
  publisher={Princeton University Press}
}

\appendix
\section{Optimizing expected Stackelberg payoff over mixed strategies} 
\label{sec:mixed-stackelberg}

The continuous relaxation of \ref{eq:ilp} is obtained by replacing the integrality constraints with non-negativity: $x_{i,j} \ge 0$, $u_i \ge 0$, and $v_\ell \ge 0$. Note that $x_{i,j} \le 1$, $u_i \le 1$, and $v_\ell \le 1$ are naturally enforced by constraints (1), (2), and the minimality of the objective.

Recall that the robotaxi assignment game is zero-sum: Firm~1 commits to a strategy to maximize riders served, while Firm~2 maximizes riders stolen. Let $\mathcal{M}_Y$ and $\mathcal{M}_Z$ denote the sets of all perfect bipartite matchings from $X$ to $Y$ and from $X$ to $Z$, respectively. A mixed strategy for Firm~1 is a probability distribution $P = (p_M)_{M \in \mathcal{M}_Y}$. For any mixed strategy $P$, the marginal assignment probabilities are given by
\[
x_{i,j} = \Pr_{M \sim P}[M(i) = j] = \sum_{M \in \mathcal{M}_Y : M(i) = j} p_M.
\]
By construction, $x = (x_{i,j})$ is a doubly stochastic matrix satisfying constraints (1) and (2) of \ref{eq:ilp}. Conversely, by the Birkhoff--von Neumann theorem~\cite{birkhoff1946tres, von1953certain}, every doubly stochastic matrix $x$ can be decomposed in polynomial time into a convex combination of at most $n^2 - 2n + 2$ permutation matrices, defining an equivalent mixed strategy $P$ whose marginal probabilities coincide with $x$. Compactly optimizing over marginal assignment probabilities and subsequently reconstructing mixed strategies via the Birkhoff--von Neumann decomposition is a powerful technique that has been successfully leveraged in Stackelberg security and resource allocation games~\cite{kiekintveld2009computing, korzhyk2010complexity}.

Suppose Firm~1 commits to a mixed strategy with marginals $x$. Observing this, Firm~2 plays a best response $M_2 \in \mathcal{M}_Z$ to maximize the expected number of stolen riders. Rider $i \in X$ is stolen if $d(i, M_2(i)) < d(i, M(i))$, which occurs when $M(i) \in D_{i, M_2(i)}$. By linearity of expectation, the expected number of riders stolen by Firm~2 under strategy $M_2$ is
\[
\mathbb{E}_{M \sim P}\left[ \sum_{i \in X} \mathbb{I}\{d(i, M_2(i)) < d(i, M(i))\} \right] 
= \sum_{i \in X} \Pr_{M \sim P}[M(i) \in D_{i, M_2(i)}] 
= \sum_{i \in X} \sum_{j \in D_{i, M_2(i)}} x_{i,j}.
\]
Firm~2's best-response problem against $x$ is therefore to compute a bipartite matching $M_2$ maximizing the quantity $\sum_{i \in X} \sum_{j \in D_{i, M_2(i)}} x_{i,j}$. Because the bipartite matching polytope is integral and all edge weights are non-negative, the maximum over pure matchings is identical to the optimum of its fractional relaxation:\footnote{Since we use the same linear programming relaxations we used in the proof of Lemma~\ref{lem:upper_bound}, we follow the same labeling here and call them again \ref{eq:lp-3} and \ref{eq:lp-2}.}
\begin{align*}
\text{maximize} \quad & \sum_{i \in X} \sum_{\ell \in Z} f_{i,\ell} \sum_{j \in D_{i,\ell}} x_{i,j} \tag{LP 3} \\
\text{subject to:} \quad & \sum_{\ell \in Z} f_{i,\ell} \le 1, \quad \forall i \in X, \\
& \sum_{i \in X} f_{i,\ell} \le 1, \quad \forall \ell \in Z, \\
& f_{i,\ell} \ge 0, \quad \forall i \in X, \ell \in Z.
\end{align*}
By strong linear programming duality, the optimal objective value of \ref{eq:lp-3} equals that of its dual:
\begin{align*}
\text{minimize} \quad & \sum_{i \in X} u_i + \sum_{\ell \in Z} v_\ell \tag{LP 2} \\
\text{subject to:} \quad & u_i + v_\ell \ge \sum_{j \in D_{i,\ell}} x_{i,j}, \quad \forall i \in X, \ell \in Z, \\
& u_i \ge 0, \quad \forall i \in X, \\
& v_\ell \ge 0, \quad \forall \ell \in Z.
\end{align*}
Subtracting this quantity from $n$, Firm~1's guaranteed expected payoff against Firm~2's best response is
\[
n - \min_{u, v} \left( \sum_{i \in X} u_i + \sum_{\ell \in Z} v_\ell \right) = \max_{u, v} \left( n - \sum_{i \in X} u_i - \sum_{\ell \in Z} v_\ell \right).
\]
Firm~1 chooses the marginal matrix $x$ that maximizes this guaranteed payoff. This joint maximization over doubly stochastic matrices $x$ and dual vectors $(u, v)$ is precisely the continuous linear programming relaxation of \ref{eq:ilp}.

\begin{theorem}
\label{prop:minimax-equivalence}
Let $(x^*, u^*, v^*)$ be an optimal solution to the continuous linear programming relaxation of \ref{eq:ilp} with objective value $\mathrm{OPT}_{\mathrm{LP}}$. Then:
\begin{enumerate}
    \item $\mathrm{OPT}_{\mathrm{LP}}$ is equal to the minimax value of the zero-sum robotaxi assignment game.
    \item The probability distribution $P^*$ obtained from the Birkhoff--von Neumann decomposition of $x^*$ is an optimal mixed Stackelberg strategy for Firm~1, guaranteeing an expected payoff of at least $\mathrm{OPT}_{\mathrm{LP}}$ against any strategy of Firm~2.
    \item Let $f^*$ be the optimal dual variables of \ref{eq:lp-3}. Any completion of $f^*$ to a doubly stochastic matrix, when decomposed via the Birkhoff--von Neumann theorem, yields a mixed strategy $Q^*$ for Firm~2 such that $(P^*, Q^*$ constitutes a mixed Nash equilibrium of the game.
\end{enumerate}
\end{theorem}

\begin{proof}
By von Neumann's Minimax Theorem~\cite{von1953certain}, every finite two-player zero-sum game satisfies
\[
\max_{P \in \Delta(\mathcal{M}_Y)} \min_{Q \in \Delta(\mathcal{M}_Z)} \mathbb{E}_{M_1 \sim P, M_2 \sim Q}[\Pi_1(M_1, M_2)] = \min_{Q \in \Delta(\mathcal{M}_Z)} \max_{P \in \Delta(\mathcal{M}_Y)} \mathbb{E}_{M_1 \sim P, M_2 \sim Q}[\Pi_1(M_1, M_2)] = V^*,
\]
where $\Pi_1$ denotes Firm~1's payoff, $\Delta(\cdot)$ denotes the simplex of mixed strategies, and $V^*$ is the unique minimax value of the game. In any two-player zero-sum game, the leader's optimal mixed Stackelberg payoff mathematically coincides with $V^*$~\cite{conitzer2006computing, vonstengel2010leadership}, and any strategy pair achieving this value forms a mixed Nash equilibrium.

As established above, Firm~1's optimal mixed Stackelberg payoff is given by
\[
\max_{x \text{ doubly stochastic}} \min_{M_2 \in \mathcal{M}_Z} \left( n - \sum_{i \in X} \sum_{j \in D_{i, M_2(i)}} x_{i,j} \right),
\]
which, by LP duality between \ref{eq:lp-2} and \ref{eq:lp-3} and the total unimodularity of bipartite matching, is identical to the optimal objective value $\mathrm{OPT}_{\mathrm{LP}}$ of the continuous relaxation of \ref{eq:ilp}. This proves item~1.

Decomposing $x^*$ via the Birkhoff--von Neumann theorem yields a distribution $P^*$ supported on at most $n^2 - 2n + 2$ matchings in $\mathcal{M}_Y$ whose marginals match $x^*$ exactly. Since Firm~2's expected payoff against $P^*$ depends strictly on these marginals, Firm~2 cannot steal more than $n - \mathrm{OPT}_{\mathrm{LP}}$ riders in expectation under any response, securing Firm~1 an expected payoff of at least $\mathrm{OPT}_{\mathrm{LP}}$. This proves item~2.

Finally, consider the dual of the linear programming relaxation of \ref{eq:ilp}. Assigning dual variables $\alpha_i$ to constraints (1), $\beta_j$ to constraints (2), and $f_{i,\ell} \ge 0$ to constraints (3), the dual LP optimizes over $f$ subject to $\sum_{\ell \in Z} f_{i,\ell} \le 1$ and $\sum_{i \in X} f_{i,\ell} \le 1$, precisely encoding Firm~2's minimax problem of minimizing Firm~1's maximum payoff over all response matchings. Because the objective value of the dual equals $\mathrm{OPT}_{\mathrm{LP}}$, the optimal matrix $f^*$ defines the marginal stealing probabilities of an optimal minimax strategy for Firm~2. Greedily augmenting $f^*$ to a doubly stochastic matrix (by pairing remaining unmatched entries with zero weight) and decomposing it via the Birkhoff--von Neumann theorem yields a mixed strategy $Q^* \in \Delta(\mathcal{M}_Z)$ that prevents Firm~1 from obtaining an expected payoff greater than $\mathrm{OPT}_{\mathrm{LP}}$. Therefore, $(P^*, Q^*)$ is a mutual minimax pair, establishing that $(P^*, Q^*)$ is a mixed Nash equilibrium.
\end{proof}
\end{document}